\documentclass[journal]{IEEEtran}

\usepackage{hyperref}
\usepackage{url}
\usepackage{amssymb,amsfonts,amsmath}
\usepackage{mathtools}

\usepackage[linesnumbered,ruled,vlined]{algorithm2e}
\usepackage{stackengine}

\usepackage{float}

\usepackage{color}
\usepackage{booktabs}

\newcommand{\N}{{\mathbb{N}}}

\newcommand{\R}{{\mathbb{R}}}

\newcommand{\vectornorm}[1]{\left|\left|#1\right|\right|}
\newcommand{\ip}[2]{\langle #1,#2 \rangle}

\DeclareMathOperator*{\argmin}{arg\,min}

\newcommand{\avec}{\mathbf{a}}

\newcommand{\x}{\mathbf{x}}
\newcommand{\y}{\mathbf{y}}
\newcommand{\z}{\mathbf{z}}

\newcommand{\A}{\mathbf{A}}
\newcommand{\U}{\mathbf{U}}
\newcommand{\vmat}{\mathbf{V}}
\newcommand{\W}{\mathbf{W}}
\newcommand{\X}{\mathbf{X}}
\newcommand{\Y}{\mathbf{Y}}

\newcommand{\eye}{\mathbf{I}}

\newtheorem{proof}{Proof}

\newtheorem{thm}{Theorem}
\newtheorem{lem}[thm]{Lemma}

\newtheorem{defn}[thm]{Definition}

\ifCLASSINFOpdf
\else
\fi
\begin{document}
\title{Deeply-Sparse Signal rePresentations (D$\text{S}^2$P)}

\author{Demba~Ba,~\IEEEmembership{Member,~IEEE}% <-this % stops a space
\thanks{D. Ba is with School of Engineering and Applied Sciences, Harvard University, Cambridge, MA (e-mail: demba@seas.harvard.edu)}% <-this % stops a space
\thanks{Manuscript received April 19, 2005; revised August 26, 2015.}}

\markboth{Journal of \LaTeX\ Class Files,~Vol.~14, No.~8, August~2015}%
{Shell \MakeLowercase{\textit{et al.}}: Bare Demo of IEEEtran.cls for IEEE Journals}



\maketitle

\begin{abstract}
%The solution to the regularized least-squares problem $\min\limits_{\x \in \R^{p+}} \frac{1}{2} \vectornorm{\y-\A\x}_2^2 + \lambda \vectornorm{\x}_1$,  assuming $\A$ is unitary, is given by the soft-thresholding operator, or ReLu in neural network parlance, applied component-wise to $\A^{\text{T}}\y$. This equivalence is at the core of 
A recent line of work shows that %sought to build a parallel between deep neural network architectures and sparse coding/recovery and estimation. Said line of work suggests, as pointed out by Papyan~\cite{papyan2017convolutional} et al., that 
a deep neural network with ReLU nonlinearities arises from a finite sequence of cascaded sparse coding models, the outputs of which, except for the last element in the cascade, are sparse and 
unobservable. That is, intermediate outputs deep in the cascade are sparse, hence the title of this manuscript. We show here, using techniques from the dictionary learning literature that, if
the measurement matrices in the cascaded sparse coding model (a) satisfy RIP and (b) all have sparse columns except for the last, they can be recovered with high probability. We propose two algorithms for this purpose: one that recovers the matrices in a forward sequence, and another that recovers them in a backward sequence.  % in the absence of noise using %O(max()) observations, and an optimization algorithm that, beginning with the last element of the cascade, alternates between estimating the dictionary and the sparse code and then, at convergence, proceeds to the preceding element in the cascade. 
The method of choice in deep learning to solve this problem is by training an auto-encoder. % whose architecture we specify. 
Our algorithms provide a sound alternative, %derived from the perspective of sparse coding, and 
with theoretical guarantees, as well
upper bounds on sample complexity. The theory shows that the learning complexity of the forward algorithm depends on the number of hidden units at the deepest layer and the number of active neurons at that layer (sparsity). In addition, the theory relates the number of hidden units in successive layers, thus giving a practical prescription for designing deep ReLU neural networks. Because it puts fewer restrictions on the architecture, the backward algorithm %, whose complexity is the the maximum, across layers, of the product of the number of active neurons (sparsity) and the number of hidden units, 
requires more data. % for recovery. %Letting $r_\ell$ be the dimension of the input of the $\ell^\text{th}$  transformation (embedding dimension) and $s_{\mathbf{Y}^{(\ell-1)}}$ the sparsity of this input (number of active neurons), the computational complexity is $\mathcal{O} \left(\text{max}_{\ell}  \text{ max}(r_\ell^2,r_\ell s_{\mathbf{Y}^{(\ell-1)}})\right) $.  
%Our proofs rely on a certain class of sparse matrices satisfying RIP. We use non-asymptotic random matrix theory to prove this for random matrices. 
We demonstrate the deep dictionary learning algorithm via simulations. 
Finally, we use a coupon-collection argument to conjecture a lower bound on sample complexity that gives some insight as to why deep networks require more data to train than shallow ones. %The simulations suggest that the term $r_\ell^2$ above is an artifact of the proof techniques we rely on to arrive at our main result. %That is, the learning complexity depends on the maximum, across layers, of the product of the number of active neurons and the embedding dimension.
\end{abstract}


% \fi
\IEEEpeerreviewmaketitle

\section{Introduction}

\IEEEPARstart{D}eep learning has been one of the most popular areas of research over the past few years, due in large part to the ability of deep neural networks  to outperform humans at a number of cognition tasks, such as object and speech recognition. 

Despite the mystique that  has surrounded their success, recent work has started to provide answers to questions pertaining, on the one hand, to basic assumptions behind deep networks--when do they work?--and, on the other hand, to interpretability--why do they work? In~\cite{patel2015probabilistic}, Patel explains deep learning from the perspective of inference in a hierarchical probabilistic graphical model. This leads to new inference algorithms based on belief propagation and its variants. In a series of articles, the authors from~\cite{papyan2017convolutional,papyan2017working,sulam2018multilayer} consider deep convolutional networks through the lens of a multi-layer convolutional sparse coding (ML-CSC) model. The authors show a correspondence between the sparse approximation step in this multi-layer model and the encoding step (forward pass) in a related deep convolutional network. Specifically, they show that convolutional neural networks with ReLU nonlinearities can be interpreted as sequential algorithms to solve for the sparse codes in the ML-CSC model. The authors carry out a detailed theoretical analysis of when the sparse recovery algorithms succeed in the absence of noise, and when they are stable in the presence of noise. More recently, building on the work of Papyan, the work in~\cite{ye2018deep} have show that some of the key operations that arise in deep learning (e.g. pooling, ReLU) can be understood from the classical theory of filter banks in signal processing. In a separate line of work, Tishby~\cite{tishby2015deep} uses the information bottleneck principle from information theory to characterize the limits of a deep network from an information-theoretic perspective.

The works~\cite{patel2015probabilistic} and~\cite{papyan2017convolutional,papyan2017working,sulam2018multilayer} relate the inference step (forward pass) of neural networks to various generative models, namely graphical models in the former and the ML-CSC model in the latter. They do not provide theoretical guarantees for learning the filters (weights) of the respective generative models. Here, we take a more expansive approach than in~\cite{patel2015probabilistic,papyan2017convolutional,ye2018deep} that connects deep networks to the theory of dictionary learning, to answer questions pertaining, not to basic assumptions and interpretability, but to the sample complexity of learning a deep network--how much data do you need to learn a deep network? 

Classical dictionary-learning theory~\cite{altmin} tackles the problem of estimating a \emph{single} unknown transformation from data obtained through a sparse coding model. The theory gives bounds for the sample complexity of learning a dictionary as a function of the parameters of the sparse-coding model. Motivated by classical dictionary learning, recent work~\cite{nguyen2019dynamics} shows that a two-layer auto-encoder with ReLU nonlinearity trained by backpropagation, using an \emph{infinite amount of data}, learns the dictionary from a sparse-coding model. %Two key features unite the works from~\cite{patel2015probabilistic,papyan2017convolutional} and~\cite{ye2018deep}. The first is (a) sparsity, and the second (b) the use of a hierarchy of transformations/representations as a proxy for the different layers in a deep neural networks. 
Neither classical dictionary learning theory, nor the work from~\cite{nguyen2019dynamics}, provide a framework for assessing the complexity of learning a hierarchy, or sequence, of transformations from data. 

We formulate a deep version of the classical sparse-coding generative model from dictionary learning~\cite{altmin}: starting with a sparse code, we apply a composition of linear transformations to generate an observation. We constrain all the transformations in the composition, except for the last, to have sparse columns, so that their composition yields sparse representations at every step. We solve the deep dictionary learning problem--learning all of the transformations in the composition--by sequential alternating minimization. We introduce two sequential algorithms. The first is a forward-factorization algorithm that learns the transformations in a forward sequence, starting with the first and ending with the last. The second is a backward-factorization algorithm that learns the transformations in a backward sequence.

% Within each algorithm, the alternating-minimization step involves a sparse approximation step, i.e. a search for a sparse input to each of the transformations in the composition. That's why, we constraint the intermediate matrices in the composition to be sparse.} As pointed out by the authors in~\cite{aberdam2019multi}, who introduce the notion of cosparsity, this a sufficient but not necessary condition for producing sparse outputs at each level. We do not consider the cosparse setting. As we detail in the main text, our notion of depth refers to the number of transformations in the composition.

We begin the rest of our treatment by briefly introducing notation (Section~\ref{sec:nottn}). Our main contributions are the following

\noindent  \textbf{The connection between dictionary learning and auto-encoders} First, in Section~\ref{sec:shallow} we develop the connection between classical dictionary learning and deep recurrent sparse auto-encoders~\cite{gregor2010learning,rolfe2013discriminative,tolooshams2018scalable,chang2019randnet}.  This motivates a deep generalization of the sparse coding model, developed in Section~\ref{sec:ds2p}, and its connection to deep ReLU auto-encoders.

\noindent \textbf{Upper bounds on the sample complexity of deep sparse dictionary learning} Second, %we use this connection in Section~\ref{sec:ds2p} to introduce a deep generalization of the sparse-coding model for dictionary learning, 
we prove in Section~\ref{sec:ds2p} that, under regularity assumptions, both the forward and backward-factorization algorithms can learn the sequence of dictionaries in the deep model. For each algorithm, we characterize the computational complexity of this learning as a function of the model parameters. Let $\{r_\ell\}_{\ell=1}^L$ be the %dimension of the input of the $\ell^\text{th}$ transformation, i.e.
the number of hidden units at layer $\ell$ (layer $L$ is the shallowest hidden layer and layer $1$ the deepest). Further  let $\{s_{\mathbf{Y}^{(\ell-1)}}\}_{\ell=1}^L$ be the %sparsity of this input, i.e. the 
number of active neurons at layer $\ell$, and $\{s_{(\ell)}\}_{\ell=1}^{L-1}$ denote the sparsity level of the sparse weight matrix connecting hidden layer $\ell+1$ to hidden layer $\ell$. The computational complexity of the forward-factorization is $\mathcal{O} \left(\text{ max}(r_1^2,r_1 s_{\Y^{(0)}} )\right)$, i.e., a function of the number of hidden units in the deepest layer and the number active neurons in that layer. In addition, the forward-factorization algorithm requires $r_{\ell} = \mathcal{O} \left( \text{ max}(r_{\ell+1}^2,r_{\ell+1} s_{(\ell)}) \right)$, $\ell=1,\ldots,L-1$, i.e., it relates the number of hidden units at a given layer to that at the preceding layer, giving a practical prescription for designing deep ReLU architectures. As detailed in the main text, the
% complexity of the
backward-factorization algorithm %is $\mathcal{O} \left(\text{max}_{\ell}  \text{ max}(r_1 r_\ell \frac{s_{\mathbf{Y}^{(\ell-1)}}}{s_{\mathbf{Y}^{(0)}}},r_\ell s_{\mathbf{Y}^{(\ell-1)}})\right) $: as detailed in the main text, it 
requires more data because it puts less stringent constraints on the architecture.

\noindent \textbf{Conjecture on a lower bound on complexity} Third, we use an argument based on the coupon-collector problem~\cite{anceaume2015new} to conjecture a lower bound $\mathcal{O}(\frac{r_1}{s_{\Y^{(0)}}} \text{ log } r_1)$ on complexity. This bound gives some insight as to why deep networks require more data to train than shallow ones.
	
	%the computational complexity is $\mathcal{O} \left(\text{max}_{\ell}  \text{ max}(r_\ell^2,r_\ell s_{\mathbf{Y}^{(\ell-1)}})\right) $. As in~\cite{altmin}, the term $r_\ell^2$ seems to be an over-estimate from the proof techniques used. This bound can be interpreted as a statement regarding the complexity of learning deep versions of the recurrent auto-encoders from Section~\ref{sec:shallow}. Indeed, in neural networks terminology, $r_\ell$ is the size of the embedding at $\ell^\text{th}$ layer and $s_{\mathbf{Y}^{(\ell-1)}}$ the number active neurons at that layer. 
	
\noindent \textbf{A characterization of the properties of column-sparse random matrices} Fourth, our proofs rely on a certain class of sparse matrices satisfying the restricted-isometry property (RIP)~\cite{candes2008restricted}. We prove this in Section~\ref{sec:sprand} for random matrices using results from non-asymptotic random matrix theory~\cite{vershynin2010introduction}.

\noindent \textbf{A notion of RIP for product matrices} Finally, our proof of convergence of the forward-factorization algorithm relies on a structured \emph{product} of matrices satisfying RIP, i.e., a notion of near-isometry for the product of matrices. We introduce such a notion and prove it for the matrices involved in the deep sparse-coding model. The RIP was originally developed for single matrices. Our result may be of independent interest.
		
	We demonstrate the algorithms for learning the sequence of transformations in the deep sparse-coding model via simulation in Section~\ref{sec:sims}. Similar to~\cite{altmin}, the simulations suggest that the second-order terms in the complexity results above are an artifact of the proof techniques we borrow from~\cite{altmin}. %, which we rely on to prove our main results. %That is, the learning complexity depends on the maximum, across layers, of the product of the number of active neurons and the embedding dimension. 
	We give concluding remarks in Section~\ref{sec:disco}.

\section{Notation}
\label{sec:nottn}

We use bold font for matrices and vectors, capital letters for matrices, and lower-case letters for vectors. For a matrix $\A$, $\avec_j$ denotes its $j^{\text{th}}$ column vector, and $\A_{ij}$ its element at the $i^{\text{th}}$ row and the $j^{\text{th}}$ column.  For a vector $\x$, $x_i$ denotes its $i^{\text{th}}$ element. $\A^\text{T}$ and $\x^\text{T}$ refer, respectively, to the transpose of the matrix $\A$ and that of the vector $\x$. We use $||\x||_p$ to denote the $\ell_p$ norm of the vector $\x$. We use $\sigma_{\text{min}}(\A)$ and $\sigma_{\text{max}}(\A)$ to refer, respectively, to the minimum and maximum singular values of the matrix $\A$. We will also use $\vectornorm{\A}_2$ to denote the spectral norm (maximum singular value of a matrix). We will make it clear from context whether a quantity is a random variable/vector. We use $\eye$ to refer the identity matrix. Its dimension will be clear from the context. Let $r \in \N^+$. For a vector $\x \in \R^r$, $\text{Supp}(\x) \subset \{1,\ldots,r\}$ refers to set of indices corresponding to its nonzero entries. 

\section{Shallow Neural Networks and Sparse Estimation}
\label{sec:shallow}

The rectifier-linear unit--ReLU--is a popular nonlinearity in the neural-networks literature. Let $z \in \R$, the ReLU nonlinearity is the scalar-valued function defined as ReLU($z) = \text{max}(z,0)$.  In this section, we build a parallel between sparse approximation/$\ell_1$-regularized regression and the ReLU($\cdot$) nonlinearity. This leads to a parallel between dictionary learning~\cite{altmin} and auto-encoder networks~\cite{rolfe2013discriminative,sreter2017learned}. In turn, this motivates an interpretation of learning a deep neural network as a hierarchical, i.e., sequential, dictionary-learning problem in a generative model that is the cascade of sparse-coding models~\cite{papyan2017convolutional}.

%we first show that the ReLu($\cdot$) nonlinearity arises in the solution to a simple constrained $\ell_1$-regularized regression problem in one dimension. This result generalizes easily to the case when the observations and the optimization variable both live in higher dimensions and are related through a unitary transform. 

%Second, we use the relationship between the ReLu($\cdot$) and $\ell_1$-regularized regression to draw a parallel between dictionary learning~\cite{} and a specific auto-encoder~\cite{} neural-network architecture. In turn, this allows us to employ dictionary learning theory to assess the computational complexity of learning the weights from this auto-encoder architecture. 

% \text{max}(y-\lambda,0) = \text{ReLu}(y-\lambda)$. , where $\text{ReLu}(y) = \text{max}(y,0)$.

\subsection{Unconstrained and non-negative $\ell_1$-regularization in one dimension}

We begin by a derivation of the soft-thresholding operator from the perspective of sparse approximation. Let $y \in \R$, $\lambda > 0$, and consider the inverse problem
\begin{equation}
	\min\limits_{x\in \R} \frac{1}{2} (y-x)^2 + \lambda |x|.
	\label{eq:uncsoft}
\end{equation}
\noindent It is well-known that the solution $\hat{x}$ to Equation~\ref{eq:uncsoft} is given by the soft-thresholding operator, i.e.,
\begin{equation}
	\hat{x} = \text{sgn}(y)\text{max}(|y|-\lambda,0) \coloneqq s_\lambda(y) = \text{ReLU}(y-\lambda)-\text{ReLU}(-y-\lambda).
	\label{eq:softop}
\end{equation}
\noindent For completeness, we give the derivation of this result in the appendix. If $x$ is non-negative, Equation~\ref{eq:uncsoft} becomes
%\noindent We show next that, subject to a non-negativity constraint, the solution to Equation~\ref{eq:uncsoft} is a translated version of ReLU~\cite{glorot2011deep} applied to $y$, a form that is more familiar to researchers from the neural-networks community. That is, the ReLU($\cdot$) nonlinearity arises in the solution to a simple constrained $\ell_1$-regularized regression problem in one dimension.  

%\subsection{\textcolor{blue}{Non-negative $\ell_1$-regularization in one dimension}}
%Let $y \in \R$, $\lambda > 0$, and c
%\noindent Consider the inverse problem
\begin{equation}
	\min\limits_{x\in \R^+} \frac{1}{2} (y-x)^2 + \lambda x,
	\label{eq:csoft}
\end{equation}
\noindent and its solution % The solution to Equation~\ref{eq:csoft} is 
$\hat{x} = \text{max}(y-\lambda,0) = \text{ReLU}(y-\lambda)$. In other words, the ReLU($\cdot$) nonlinearity arises in the solution to a simple constrained $\ell_1$-regularized regression problem in one dimension. To see this, for $y > 0$, the solution coincides with that of Equation~\ref{eq:uncsoft}. For $y \leq 0$, the solution must be $\hat{x} = 0$. Suppose that $x > 0$, then the value of the objective function is $\frac{1}{2} (y-x)^2 + \lambda x$, which is strictly greater that $\frac{1}{2} y^2$, i.e., the objective function evaluated at $x = 0$. 

\subsection{Unconstrained $\ell_1$-regularization in more than one dimension}

The above results generalize easily to the case when the observations and the optimization variable both live in higher dimensions and are related through a unitary transform. Let $r \in \N^+$ and $\y \in \R^{r}$, $\lambda > 0$, and $\A$ be a unitary matrix. Consider the problem
\begin{equation}
	\min\limits_{\x\in \R^{r+}} \frac{1}{2} \vectornorm{\y-\A\x}_2^2 + \lambda \vectornorm{\x}_1.
	\label{eq:pdcsoft}
\end{equation}
\noindent Since $\A$ is unitary, i.e., an isometry, Equation~\ref{eq:pdcsoft} is equivalent to
\begin{equation}
	\min\limits_{\x\in \R^{r+}} \frac{1}{2} \vectornorm{\tilde{\y}-\x}_2^2 + \lambda \vectornorm{\x}_1 = \min\limits_{\x\in \R^{r+}} \sum_{j=1}^r \frac{1}{2} (\tilde{y}_j - x_j)^2 + \lambda x_j,
	\label{eq:pdcsoftiso}
\end{equation}
\noindent where $\tilde{y}_j = \avec_j^{\text{T}} \y$, $j = 1,\ldots,r$. 
\noindent Equation~\ref{eq:pdcsoftiso} is separable in $x_1,\ldots,x_r$. For each $x_j$, the optimization is equivalent to Equation~\ref{eq:csoft} with $\tilde{y}_j$ as the input, $j = 1,\ldots,r$. Therefore,
%\begin{eqnarray}
%	\hat{x}_j & = & \text{max}(\tilde{y_j} - \lambda,0) \\ 
%	& = & \left\{ \begin{array}{ll}
%        \avec_j^{\text{T}} \y - \lambda  & \mbox{if $ \avec_j^{\text{T}} \y > \lambda$};\\
%       0 & \mbox{\text{otherwise}}.\end{array} \right.\\
%       & = & \left\{ \begin{array}{ll}
%        \begin{bmatrix} \avec_j \\ -\lambda \end{bmatrix}^{\text{T}} \begin{bmatrix} \y \\ 1 \end{bmatrix}  & \mbox{if $ \begin{bmatrix} \avec_i \\ -\lambda \end{bmatrix}^{\text{T}} \begin{bmatrix} \y \\ 1 \end{bmatrix} > 0$};\\
%       0 & \mbox{\text{otherwise}}.\end{array} \right. \\
%       & = & \text{ReLU}\left(\begin{bmatrix} \avec_j \\ -\lambda \end{bmatrix}^{\text{T}} \begin{bmatrix} \y \\ 1 \end{bmatrix}\right).
%       \label{eq:softrelueq}
%\end{eqnarray}
\begin{equation}
      \hat{x}_j  = \text{ReLU}\left(\begin{bmatrix} \avec_j \\ -\lambda \end{bmatrix}^{\text{T}} \begin{bmatrix} \y \\ 1 \end{bmatrix}\right).
       \label{eq:softrelueq}
\end{equation}

\begin{figure}[H]
  \centering
  \includegraphics{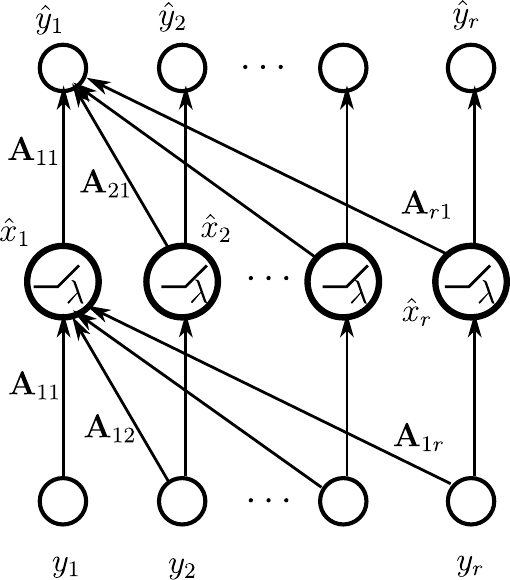}
\caption{Auto-encoder architecture motivated by the interpretation of the ReLU in the context of sparse approximation. The encoder solves the $\ell_1$-regularized least-squares problem with non-negativity constraints. The decoder reconstructs the observations by applying the dictionary to the output of the encoder. Assuming the biases are known, the only parameter of this architecture is the dictionary $\A$, which can be learned by backpropagation. To simplify the figure, we only draw a subset of all of the connections.}
\label{fig:unitary_ae}
\end{figure}
 
\noindent %Equation~\ref{eq:softrelueq} states that, for $\A$ unitary, the solution to the $\ell_1$-regularized least-squares problem with non-negativity constraints (Equation~\ref{eq:pdcsoft}) is obtained component-wise, by  projecting the vector $ \begin{bmatrix} \y \\ 1 \end{bmatrix}$ onto the vector $\begin{bmatrix} \avec_j \\ -\lambda \end{bmatrix}^{\text{T}}$ and passing it through the ReLU($\cdot$). nonlinearity. 
Equation~\ref{eq:softrelueq} states that, for $\A$ unitary, a simple one-layer feed-forward neural network solves the inverse problem of Equation~\ref{eq:pdcsoft}, and that %Equation~\ref{eq:softrelueq} also suggests that 
$-\lambda$ plays the role of the bias in neural networks. %Allowing for different biases is akin to using a different regularization parameter for each of the components of $\x$. 
Applying the transformation $\A$ to the vector $\hat{\x}$ yields an approximate reconstruction $\hat{\y} = \A \hat{\x}$. We depict this two-stage process as a two-layer feed-forward neural network in Figure~\ref{fig:unitary_ae}. The architecture depicted in the figure is called an auto-encoder~\cite{rolfe2013discriminative,sreter2017learned,tolooshams2018scalable,chang2019randnet}. Given training examples, the weights of the network, which depend on $\A$, can be learned by backpropagation. This suggests a connection between dictionary learning and auto-encoder architectures, which we elaborate upon below.

\noindent \underline{\textbf{Remark 1}}: The literature suggests that the parallel between the ReLU and sparse approximation dates to the work~\cite{gregor2010learning}. %Prior to this, while they do not explicitly make this connection, 
In~\cite{glorot2011deep}, the authors discuss in detail the sparsity-promoting properties of the ReLU compared to other nonlinearities in neural networks.

\subsection{Sparse coding, dictionary learning, and auto-encoders}

We use the relationship between ReLU and $\ell_1$-regularized regression to draw a parallel between dictionary learning~\cite{altmin} and a specific auto-encoder neural-network architecture~\cite{rolfe2013discriminative,sreter2017learned,tolooshams2018scalable,chang2019randnet}. %We keep the treatment informal here. \textcolor{blue}{We will formally introduce notation in Section~\ref{sec:}}.

\noindent \paragraph{Shallow sparse generative model} Let $\mathbf{Y}$ be a $d \times n$ real-valued matrix generated as follows
\begin{equation}
	\mathbf{Y} = \mathbf{A}\mathbf{X}, \mathbf{A} \in \R^{d \times r}, \mathbf{X} \in \R^{r \times n}.
	\label{eq:shallowdl}
\end{equation}
\noindent Each column of $\mathbf{X}$ is an $s$-sparse vector, i.e., only $s$ of its elements are non-zero~\cite{altmin}. The non-zero elements represent the coordinates or codes for the corresponding column of $\mathbf{Y}$ in the dictionary $\mathbf{A}$~\cite{altmin}. 

\noindent \underline{\textbf{Remark 2}}: We call this model ``shallow'' because there is only one transformation $\A$ to learn. %In Section~\ref{sec:ds2p}, we will contrast this with a ``deep'' generative model where we will learn each of the transformations that comprise the composition of multiple linear transformations applied to a sparse code. 

\noindent \paragraph{Sparse coding and dictionary learning} Given $\mathbf{Y}$, the goal is to estimate $\mathbf{A}$ and $\mathbf{X}$. Alternating minimization~\cite{altmin} (Algorithm~\ref{algo:altminalgo}) is a popular algorithm to find $\mathbf{A}$ and $\mathbf{X}$ as the solution to
\begin{eqnarray}
	(\hat{\A},\hat{\mathbf{X}}) & = & \argmin \limits_{\X,\A} \vectornorm{\x_i}_1, \forall i=1,\ldots,n  \nonumber \\
	& &\text{ s.t. } \mathbf{Y} = \A\mathbf{X}, \vectornorm{\avec_j}_2 = 1, \nonumber \\
	& & \forall j=1,\ldots,r.
	\label{eq:scopt}
\end{eqnarray}
\noindent Algorithm~\ref{algo:altminalgo} solves Equation~\ref{eq:scopt} by alternating between a sparse coding step, which updates the sparse codes given an estimate of the dictionary, and a dictionary update step, which updates the dictionary given estimates of the sparse codes. In unconstrained form, the sparse-coding step of the $t^{\text{th}}$ iteration of is equivalent to solving Equation~\ref{eq:pdcsoft}, with a value for $\lambda$ that depends on $\epsilon_t$. We defer a discussion of the assumptions behind Algorithm~\ref{algo:altminalgo} to Section~\ref{sec:ds2p}, where we introduce two deep generalizations of the algorithm.
\begin{algorithm}
%\DontPrintSemicolon % Some LaTeX compilers require you to use \dontprintsemicolon    instead
\KwIn{Samples $\mathbf{Y}$, initial dictionary estimate $\A(0)$, accuracy sequence $\epsilon_t$, sparsity level $s$, and number of iterations $T$. %Thresholding function $\text{T}_\rho(a) = a$ if $|a| > \rho$ and $0$ o.w.
}
\For{$t=0$ to $T-1$}{
	\For{$i=1$ to $n$}{
		$\mathbf{X}(t+1)_i = \argmin \limits_{\x} \vectornorm{\x}_1 \text{ s.t. } \vectornorm{\y_i-\A(t)\x}_2 \leq \epsilon_t$
	}
	Threshold: $\mathbf{X}(t+1) = \mathbf{X}(t+1) \cdot (\mathbb{I}[\mathbf{X}(t+1) > 9s\epsilon_t])$\\
	Estimate $\A(t+1) = \mathbf{Y}\mathbf{X}(t+1)^{+}$\\
	Normalize: $\A(t+1)_i = \frac{\avec(t+1)_i}{\vectornorm{\avec(t+1)_i}_2} $
}
\KwOut{$\left(\A(T),\mathbf{X}(T)\right)$}
\caption{AltMinDict($\mathbf{Y}$,$\A(0)$,$\epsilon_t$,$s$,$T$): Alternating minimization algorithm for dictionary learning}
\label{algo:altminalgo}
\end{algorithm}

Suppose that instead of requiring equality in Equation~\ref{eq:shallowdl}, our goal where instead to solve the following problem
\begin{equation}
	\min\limits_{\mathbf{A},\mathbf{X}} \frac{1}{2} \vectornorm{\mathbf{Y}-\A\mathbf{X}}_\text{F}^2 + \lambda \sum_{i=1}^n \vectornorm{\x_i}_1.
	\label{eq:shallowdlpb}
\end{equation}

If $\A$ were a unitary matrix, the sparse-coding step could be solved exactly using Equation~\ref{eq:softrelueq}. The goal of the dictionary-learning step is to minimize the reconstruction error between $\A$ applied to the sparse codes, and the observations. In the neural-network literature, this two-stage process describes so-called auto-encoder architectures~\cite{rolfe2013discriminative,sreter2017learned}.

\noindent \underline{\textbf{Remark 3}}: We make the assumption that $\A$ is unitary to simplify the discussion and make the parallel between neural networks and dictionary learning more apparent. If $\A$ is not unitary, we can use %replace Equation~\ref{eq:shallowdl} with 
the iterative soft-thresholding algorithm (ISTA)~\cite{beck2009fast}: $\x_k = s_{\frac{\lambda}{M}}\left(\x_{k-1} + \frac{1}{M}\A^{\text{T}}(\y-\A \x_{k-1})\right)$, where $k$ indexes the iterations of the algorithm, and $M \geq \sigma_{\text{max}}(\A^\text{T}\A)$.

\begin{figure}[H]
  \centering
  \includegraphics[scale=0.9]{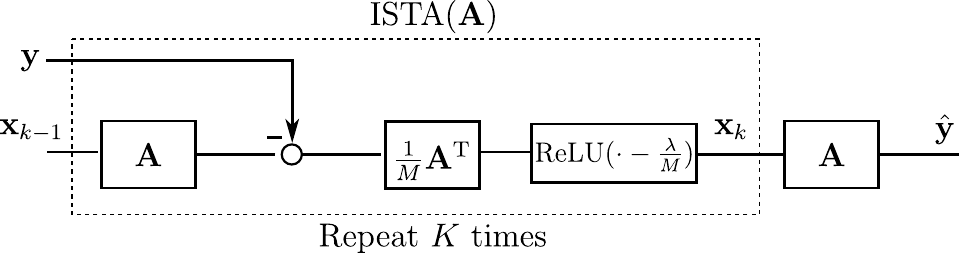}
\caption{Auto-encoder architecture motivated by alternating-minimization algorithm for sparse coding and dictionary learning. The encoder uses $K$ ($K$ large) iterations of the ISTA algorithm for sparse coding, starting with a guess $\mathbf{x}_0$ of the sparse code. The decoder reconstructs the observations by applying the dictionary to the output of the encoder. Assuming the biases are known, the only parameter of this architecture is the dictionary $\A$, which can be learned by backpropagation. $M$ is a constant such that $M \geq \sigma_{\text{max}}(\A^\text{T}\A)$.}
\label{fig:ista_ae}
\end{figure}

\noindent \paragraph{Shallow, constrained, recurrent, sparse auto-encoders.}

We introduce an auto-encoder architecture for learning the model from Equation~\ref{eq:shallowdl}. This auto-encoder has an implicit connection with the
alternating-minimization algorithm applied to the same model. Given $\A$, the encoder produces a sparse code using a finite
(large) number of iterations of the ISTA algorithm~\cite{beck2009fast}. The
decoder applies $\A$ to the output of the decoder to reconstruct
$\y$. We call this architecture a constrained recurrent sparse
auto-encoder (CRsAE)~\cite{tolooshams2018scalable,TolooshamsBahareh2019deepresidualAE}. The constraint comes from the fact
that the operations used by the encoder and the decoder are
tied to each other through $\A$. Hence, the encoder and decoder
are not independent, unlike in~\cite{rolfe2013discriminative,sreter2017learned}. The encoder is recurrent because of the ISTA algorithm,
which is an iterative procedure. Figure~\ref{fig:ista_ae} depicts this architecture. Recent work~\cite{nguyen2019dynamics} shows that a version of this auto-encoder, with on iteration of ISTA, trained by backpropagation, learns the dictionary from a sparse-coding model under certain regularity assumptions. Their work, however, assumes an \emph{infinite amount of data} and does not consider, as we do, the problem of learning a hierarchy of dictionaries and its finite-sample complexity.

\noindent There are two basic take-aways from the previous discussion
\begin{enumerate}
	\item Constrained auto-encoders with ReLU nonlinearities capture the essence of the alternating-minimization algorithm for dictionary learning.
	\item Therefore, the sample complexity of dictionary learning can give us insights on the hardness of learning neural networks.
\end{enumerate}

\noindent \paragraph{How to use dictionary learning to assess the sample complexity of learning deep networks?} The depth of a neural network refers to the number of its hidden layers, excluding the output layer. A shallow network is one with two or three hidden layers~\cite{theodoridis2015machine}. A network with more than three hidden layers is typically called ``deep''. Using this definition, the architecture from Figure~\ref{fig:ista_ae} would be called deep. This is because of $K$ iterations of ISTA which, when unrolled~\cite{gregor2010learning,rolfe2013discriminative,tolooshams2018scalable,chang2019randnet} would constitute $K$ separate layers. This definition, however, does not reflect the fact that the only unknown in the network is $\A$. Therefore, the number of parameters of the network is the same as that in a one-layer, fully-connected, feed-forward network.

A popular interpretation of deep neural networks is that they learn a hierarchy, or sequence, of transformations of data. Motivated by this interpretation, we define the \emph{depth of a network}, not in relationship to its number of layers, but \emph{as the number of underlying distinct transformations/mappings to be learned.}

Classic dictionary learning tackles the problem of estimating a \emph{single} transformation from data~\cite{altmin}. Dictionary-learning theory characterizes the sample complexity of learning the model of Equation~\ref{eq:shallowdl} under various assumptions. We can use these results to get insights on the complexity of learning the parameters of the auto-encoder from Figure~\ref{fig:ista_ae}. Classical dictionary learning theory does not, however, provide a framework for assessing the complexity of learning a hierarchy, or sequence, of transformations from data.

\section{Deep Sparse Signal Representations}
\label{sec:ds2p}

Our goal is to build a deep (in the sense defined previously) version of the model from Equation~\ref{eq:shallowdl}, i.e., a generative model in which, starting with a sparse code, a composition of linear transformations are applied to generate an observation. What properties should such a model obey? In the previous section, we used the sparse generative model of Equation~\ref{eq:shallowdl} to motivate the auto-encoder architecture of Figure~\ref{fig:ista_ae}. \emph{The goal of the encoder is to produce sparse codes}~\cite{glorot2011deep}. We will construct a deep version of the auto-encoder and use it to infer desirable properties of the deep generative model.

\subsection{Deep, constrained, recurrent, sparse auto-encoders}

For simplicity, let us consider the case of two linear transformations $\A^{(1)}$ and $\A^{(2)}$. $\A^{(1)}$ is applied to a sparse code $\x$, and $\A^{(2)}$ to its output to generate an observation $\y$. Applied to $\y$, the goal of the ISTA$(\A^{(2)})$ encoder is to produce sparse codes. One scenario when this will happen is if $\A^{(1)}\x$, i.e., the image of $\A^{(1)}$ is sparse/approximately sparse.  
% We depict this architecture in Figure~\ref{}.
Under this scenario, The ISTA$(\A^{(1)})$ encoder would then use the sparse output of the ISTA$(\A^{(2)})$ encoder to produce an estimate $\hat{\x}$ of $\x$, completing the steps of the encoder. The decoder a then approximates $\y$ with $\hat{\y} = \A^{(2)}\A^{(1)}\hat{\x}$, leading to a two-layer generalization of the architecture from Figure~\ref{fig:unitary_ae}.

For the composition of more than two transformations, the requirement that the encoders applied in cascade produce sparse codes suggests that, starting with a sparse code, the output of each of the transformations, expect for the very last which gives the observations, must be approximately sparse. As pointed out by the authors in~\cite{aberdam2019multi}, this is a sufficient but not necessary condition. Different from our synthesis perspective, the authors take an analysis perspective that relies on the notion of cosparsity~\cite{nam2013cosparse}. We do not consider the cosparse setting.

We are in a position to specify our deep sparse generative model, along with the assumptions that accompany the model.

\subsection{Deep sparse generative model and coding}

Let $\mathbf{Y}$ be the $d_L \times n$ real-valued matrix obtained by applying the composition of $L$ linear transformations $\{\A^{(\ell)}\}_{\ell=1}^L$ to a matrix $\X$ of sparse codes
\begin{align}
	&  \mathbf{Y} = \A^{(L)}\cdots \A^{(2)}\A^{(1)} \mathbf{X}, 	\label{eq:deepdl}
 \\
	&  \mathbf{X} \in \R^{r_1 \times n}, \A^{(\ell)} \in \R^{d_\ell \times r_\ell}; \nonumber \\
	&  \forall \ell=1,\ldots,L-1, \text{all columns of }\A^{(\ell)} \text{ are }s_{(\ell)} \text{ sparse,}  \nonumber\\
	&  \text{and its nonzero entries uniformly bounded.} \nonumber
\end{align}
If we further assume that each column of $\mathbf{X}$ is $s$-sparse, i.e., at most $s$ of the entries of each column are nonzero, the image of each of the successive transformations $\{\A^{(\ell)}\}_{\ell=1}^{L-1}$ will also be sparse. Finally, we apply the transformation $\A^{(L)}$ to obtain the observations $\mathbf{Y}$. We let $\A^{(0)} = \X$, and $s = s_{(0)}$.

\noindent Given $\mathbf{Y}$, we would like solve the following problem
\begin{align}
	\min \limits_{\X,\{\A^{(\ell)}\}_{\ell=1}^L} & \vectornorm{\x_i}_1, \forall i=1,\ldots,n \nonumber\\
	& \text{ s.t. }  \mathbf{Y} = \A^{(L)}\cdots \A^{(2)}\A^{(1)} \mathbf{X}, \nonumber \vectornorm{\avec^{(\ell)}_j}_2 = 1, \nonumber \\
	& \forall j=1,\ldots,r_\ell; \forall \ell=1,\ldots,L
	\label{eq:deepdlpb}.
\end{align}

\subsubsection{Connection with shallow dictionary learning} If $\A^{(2)} = \A^{(3)} = \cdots = \A^{(L)} = \eye$, the deep sparse generative model of Equation~\ref{eq:deepdl} becomes the shallow sparse generative model of Equation~\ref{eq:shallowdl}. Therefore, the deep dictionary-learning problem of Equation~\ref{eq:deepdlpb} reduces to the shallow dictionary learning problem of Equation~\ref{eq:scopt}, a problem that is well-studied in the dictionary-learning literature~\cite{altmin}, and for which the authors from~\cite{altmin} propose an alternating-minimization procedure, Algorithm~\ref{algo:altminalgo}. In~\cite{altmin}, the authors show that, under certain assumptions, the output $\A(T)$ of Algorithm~\ref{algo:altminalgo} converges to $\A$. The three key assumptions for convergence are that $\A$ satisfies the restricted isometry property (RIP)~\cite{candes2008restricted}, the initialization $\A(0)$ is close to $\A$ (in a precise sense), and that $\{\epsilon_t\}$ forms a decreasing sequence.

\subsubsection{Reduction to a sequence of shallow problems} In the next section, we introduce two algorithms to learn the dictionaries $\{\A^{(\ell)}\}_{\ell=1}^L$ given $\Y$, by sequential application of Algorithm~\ref{algo:altminalgo}. We call the first algorithm the ``forward-factorization'' algorithm because it learns the dictionaries in a forward sequence, i.e., beginning with $\A^{(1)}$ and ending with $\A^{(L)}$. The second algorithm, termed ``backward-factorization'' algorithm, learns the dictionaries in a backward sequence.

In what follows, it will be useful to define two sets of matrices whose properties will be of interest in our theoretical analyses of the algorithms we introduce for deep dictionary learning. The first set of matrices comprises
\begin{equation}
	\A^{(\bar{\ell} \rightarrow L)} = \prod_{\ell=\bar{\ell}}^L \A^{(\ell)}, \forall \bar{\ell} = 1,\ldots,L.
%	\Y & = & \A^{(\bar{\ell}\rightarrow L)}\Y^{(\bar{\ell}-1)}, \forall \bar{\ell} = 1,\ldots,L.\\
%	\A^{(0 \rightarrow L)} & = & \Y.
\end{equation}
\noindent Given a layer index $\bar{\ell} = 1,\ldots,L$, $\A^{(\bar{\ell} \rightarrow L)}$ is the product all dictionaries from later $\bar{\ell}$ to layer $L$. In this notation, $\A^{(L \rightarrow L)} = \A^{(L)}$. The second set of matrices comprises
\begin{equation}
\mathbf{Y}^{(\ell)} = \prod_{\ell'=1}^\ell \A^{(\ell)}\X, \ell=1,\ldots,L-1.
\end{equation}
\noindent Given a layer index $\ell=1,\ldots,L-1$, $\mathbf{Y}^{(\ell)}$ is the output of the $\ell^\text{th}$ operator in Equation~\ref{eq:deepdl}. At depth $\ell$, the columns of $\mathbf{Y}^{(\ell)}$, $\ell=1,\ldots,L-1$ are sparse representations of the signal $\Y$, i.e., they are \emph{deeply sparse}. Each column of $\mathbf{Y}^{(\ell)}$ represent the hidden units/neurons at layer $\ell$ of the corresponding example, and its nonzero entries the active units/neurons. We let $\Y^{(0)} = \X$.

%Before stating the result, we introduce some notation and present the alternate algorithm. We let $\A^{(0)} = \Y^{(0)} = \X$ and

%\begin{table}[!ht]
%	\centering
%	\textcolor{blue}{
%	\begin{tabular}{c l}
%		\toprule
%		Symbol & Description\\
%		\midrule
%		\\[-0.6em]
%		$\mathbf{H}$ & Matrix\\
%		$\mathbf{h}$ & Vector\\
%		$\mathcal{S}$ & Set\\
%		$\mathbf{H}_i$ & $i^{\text{th}}$ column from $\mathbf{H}$\\
%		$\mathbf{H}^c$ & $c^{\text{th}}$ block from $\mathbf{H}$\\
%		$\mathbf{h}[j]$  & $j^{\text{th}}$ entry from $\mathbf{h}$\\
%		$\mathcal{S}_i$ & $i^{\text{th}}$ element from set $\mathcal{S}$\\
%		$\mathcal{S}^j$ & $j^{\text{th}}$ set\\
%		$\mathbf{I}_{L\times L}$ & Identity matrix of size $L\times L$\\
%		$\mathbf{r}^{(t)}$ & $\mathbf{r}$ at $t^{\text{th}}$ iteration\\
%		$\mathbf{0}_L$ & a length-$L$ vector with all entries equal to $0$\\
%		$n^c_{j,i}$ & $i^{\text{th}}$ event from source $c$ in $j^{\text{th}}$ window\\
%		$N_j^c$ & Number of events from source $c$ in $j^{\text{th}}$ window\\
%		$\lVert \cdot\rVert_p$ & $\ell_p$ norm\\
%		\\[-0.6em]
%		\bottomrule
%	\end{tabular}	
%	\caption{Notational conventions.}
%	}
%	\label{table:notation}
%\end{table}

\subsection{Learning the deep sparse coding model by sequential alternating minimization}

\subsubsection{Forward-factorization algorithm}

Algorithm~\ref{algo:deepdlalgo2} specifies the steps of the forward-factorization algorithm. The procedure relies on the sequential application of Algorithm~\ref{algo:altminalgo}.

To gain some intuition, let us consider the case of $L=3$.  We first learn the product $\A^{(1 \rightarrow 3)} = \A^{(3)}\A^{(2)}\A^{(1)}$. Having learned this product, we then use it to learn $\A^{(2 \rightarrow 3)} = \A^{(3)}\A^{(2)}$, which automatically yields $\A^{(1)}$. Finally, we use $\A^{(3)}\A^{(2)}$ to learn $\A^{(2)}$ and $\A^{(3)}$. The sequence in which the forward-factorization algorithm learns the dictionaries is similar \emph{in spirit} to the multi-layer sparse approximation algorithm (Algorithm 1) from~\cite{le2015chasing}. For the dictionary update step,~\cite{le2015chasing} uses a projected gradient algorithm to enforce the constraints on the columns of the dictionaries, while we use least-squares, followed by normalization.

\begin{algorithm}
%\DontPrintSemicolon % Some LaTeX compilers require you to use \dontprintsemicolon    instead
\KwIn{Samples $\mathbf{Y}$, number of levels $1 \leq \bar{\ell} \leq L$, initial dictionary estimates $\{\A^{(\ell \rightarrow L)}(0)\}_{\ell=1}^{\bar{\ell}}$, accuracy sequences $\{\epsilon^{(\ell \rightarrow L)}_t\}_{\ell=1}^{\bar{\ell}}$, sparsity levels $\{s_{(\ell-1)}\}_{\ell=1}^{\bar{\ell}}$. %Thresholding function $\text{T}_\rho(a) = a$ if $|a| > \rho$ and $0$ o.w.
}
$\hat{\A}^{(0 \rightarrow L)} = \mathbf{Y}$\\
	$\left(\hat{\A}^{(1 \rightarrow L)},\hat{\X}\right) = \text{AltMinDict}(\hat{\A}^{(0 \rightarrow L)},\A^{(1 \rightarrow L)}(0),\epsilon^{(1 \rightarrow L)}_t,s,\infty)$\\
$\ell = 2$\\
\While{$\ell \leq \bar{\ell}$}{
	$\left(\hat{\A}^{(\ell \rightarrow L)},\hat{\A}^{(\ell-1)}\right) = \text{AltMinDict}(\sqrt{s_{(\ell-1)}}\hat{\A}^{(\ell-1 \rightarrow L)},\A^{(\ell \rightarrow L)}(0),\epsilon^{(\ell \rightarrow L)}_t,s_{(\ell-1)},\infty)$\\
	$\ell = \ell + 1$
}
\KwOut{$\{(\hat{\A}^{(\ell \rightarrow L)},\hat{\A}^{(\ell-1)})\}_{\ell=1}^{\bar{\ell}}$}
\caption{Forward-factorization algorithm for deep dictionary learning}
\label{algo:deepdlalgo2}
\end{algorithm}

We now state explicitly assumptions on the deep generative model of Equation~\ref{eq:deepdl} that will guarantee the success of Algorithm~\ref{algo:deepdlalgo2} for arbitrary $L$. % and sample-complexity estimates for the success, for arbitrary $L$, of Algorithm~\ref{algo:deepdlalgo2}. 
The reader can compare these assumptions to assumptions A1--A7 from~\cite{altmin}. %As in~\cite{altmin}, we assume, without any loss in generality that the columns of $\{\A^{(\ell)}\}_{\ell=1}^L$ all have unit $\ell_2$ norm, i.e. $\vectornorm{\avec_j^{(\ell)}}_2 = 1$, $j=1,\cdots,r_\ell$, $\ell = 1,\cdots,L$.

\begin{itemize}
	\item[(A1)] \textbf{Product Dictionary Matrices satisfying RIP}: For each $\ell=1,\ldots,L$, the product dictionary matrix $\A^{(\ell \rightarrow L)}$ has $2s_{(\ell-1)}$-RIP constant $\delta_{2s_{(\ell-1)}} < 0.1$.
	\item[(A2)] \textbf{Spectral Condition of Product Dictionary Matrices:} For each $\ell=1,\ldots,L$, the dictionary matrix $\A^{(\ell \rightarrow L)}$ has bounded spectral norm, i.e., for some constant $\mu_{(\ell \rightarrow L)} > 0$, $\vectornorm{\A^{(\ell \rightarrow L)}}_2 < \mu_{(\ell \rightarrow L)} \sqrt{\frac{r_\ell}{d_L}}$.
	\item[(A3a)] \textbf{Non-zero Entries of $\X$}: The non-zero entries of $\mathbf{X}$ are drawn i.i.d. from a distribution such that $\mathbb{E}[(\mathbf{X}_{ij})^2] = 1$, and satisfy the following a.s.: $|\mathbf{X}_{ij}| \leq M_{(0)}, \forall, i,j$.
	\item[(A3b)] \textbf{Non-zero Entries of $\{\A^{(\ell)}\}_{\ell=1}^{L-1}$}: $\forall \ell=1,\ldots,L-1$, the non-zero entries of $\A^{(\ell)}$ are drawn i.i.d. from a distribution with mean equal to zero, such that $\mathbb{E}[(\A^{(\ell)}_{ij})^2] = \frac{1}{s_{(\ell)}}$, and satisfy the following a.s.: $|\sqrt{s_{(\ell)}}\A^{(\ell)}_{ij}| \leq M_{(\ell)}, \forall, i,j$.	
	\item[(A4)] \textbf{Sparse Coefficient Matrix}: Recalling that $\X = \A^{(0)}$, $\forall \ell=0,\ldots,L-1$, the columns of $\A^{(\ell)}$ have $s_{(\ell)}$ non-zero entries which are selected uniformly at random from the set of all $s_{(\ell)}$-sized subsets of $\{1,\ldots,r_{\ell+1}\}$. We further require that, for $\forall \ell=0,\ldots,L-1$, $s_{(\ell)} \leq \frac{d_L^{1/6}}{c_2 {\mu_{(\ell+1 \rightarrow L)}}^{1/3}}$.
	\item[(A5a)] \textbf{Sample Complexity}: Given the failure parameter $\delta_{(1 \rightarrow L)} > 0$, the number of samples $n$ needs to satisfy
	\begin{equation}
	n = \mathcal{O} \left(\text{ max}(r_1^2,r_1 M^2_{(0)} s) \text{log}\left(\frac{2 r_1 }{\delta_{(1 \rightarrow L)}}\right)\right).
		\label{eq:cmpfwd}
	\end{equation}
	\item[(A5b)] \textbf{Scaling Law for Number of Hidden Units}: Given the failure parameters $\{\delta_{(\ell \rightarrow L)}\}_{\ell=2}^L > 0$, the size of the hidden layers $r_\ell$ needs to satisfy, $\forall \ell = 1,\ldots,L-1$	
	\begin{equation}
		r_{\ell} = \mathcal{O} \left(\text{ max}(r_{\ell+1}^2,r_{\ell+1} M^2_{(\ell)} s_{(\ell)}) \text{log}\left(\frac{2 r_{\ell+1} }{\delta_{(\ell+1 \rightarrow L)}}\right)\right).
	\end{equation}
	\item[(A6)] \textbf{Initial Dictionary with Guaranteed Error Bound}: It is assumed that, $\forall \ell = 1,\ldots,L$, we have access to an initial dictionary estimate $\A^{(\ell \rightarrow L)}$ such that 
	\begin{equation}
	\max_{i}\min_{z \in \{-1,1\}} \vectornorm{ z \A_i^{(\ell \rightarrow L)}(0)-\A_i^{(\ell \rightarrow L)} }_2 \leq \frac{1}{2592 s_{(\ell-1)}^2}.
	\end{equation}
	\item[(A7)] \textbf{Choice of Parameters for Alternating Minimization}: For all $\ell = 1,\ldots,L$, line 5 of Algorithm~\ref{algo:deepdlalgo2} uses a sequence of accuracy parameters $\epsilon^{(\ell \rightarrow L)}_0 = \frac{1}{2592s_{(\ell-1)}^2}$ and
	\begin{equation}
	\epsilon^{(\ell \rightarrow L)}_{t+1} = \frac{25050\mu_{(\ell \rightarrow L)} s_{(\ell-1)}^3}{\sqrt{d_L}}\epsilon^{(\ell \rightarrow L)}_t.
	\end{equation}
\end{itemize}

\noindent Comparing these assumptions to  A1--A7 from~\cite{altmin}, the least trivial one is assumption A1, namely that the product matrices $\{\A^{(\ell \rightarrow)}\}_{\ell=1}^L$ satisfy RIP~\cite{candes2008restricted}. The RIP was introduced for single matrices and not products of matrices. One of our contributions is to show that if the sparse matrices $\A^{(1)},\ldots,\A^{(L-1)}$, as well as the (potentially) dense matrix $\A^{(L)}$ satisfy RIP of some order, then the product matrices will satisfy RIP. Formally, 

\begin{thm}[RIP-like property of $\A^{(\bar{\ell} \rightarrow L)}$]
	Suppose that for each $\ell=1,\ldots,L$, the dictionary matrix $\A^{(\ell)}$ satisfies RIP of order $2s_{\mathbf{Y}^{(\ell-1)}}$ with  constant $\delta_{2s_{\mathbf{Y}^{(\ell-1)}}}$. Then $\forall \bar{\ell}=1,\ldots,L$ and $2s_{\mathbf{Y}^{(\bar{\ell}-1)}}$-sparse $\y^{(\bar{\ell}-1)}$, 
	\begin{eqnarray}
		& & \prod_{\ell=\bar{\ell}}^L (1-\delta_{2s_{\mathbf{Y}^{(\ell-1)}}}) \vectornorm{\y^{(\bar{\ell}-1)}}_2^2 
		\leq \nonumber \\ 
		\hspace{-0.25in} \vectornorm{\A^{(\bar{\ell} \rightarrow L)}\y^{(\bar{\ell}-1)}}_2^2 
		 & \leq & \prod_{\ell=\bar{\ell}}^L (1 +\delta_{2s_{\mathbf{Y}^{(\ell-1)}}}) \vectornorm{\y^{(\bar{\ell}-1)}}_2^2.
	\end{eqnarray}	
	\label{thm:prodrip}
\end{thm}

We prove the theorem, and make further remarks in the appendix. 

The theorem assumes that there exists matrices whose columns have a \emph{fixed} sparsity level, and hence dependent entries, that satisfy RIP (assumption A1). In Section~\ref{sec:sprand}, we show that, with high probability, sparse random matrices that obey assumption A3b and A4 satisfy RIP. We will appeal to standard results from random matrix theory~\cite{vershynin2010introduction}. From the analysis in Section~\ref{sec:sprand}, assumption A2 follows from the concentration of the singular values of such matrices. %Assumption A3b comes from the assumptions in Section~\ref{sec:sprand} on class of sparse random matrices that we show satisfy RIP with high the  the analysis 

We now state our main result on the ability of the forward-factorization algorithm, Algorithm~\ref{algo:deepdlalgo2}, to recover the sequence of dictionaries from the deep sparse generative model.

\begin{thm}[Exact recovery of the deep generative model by Algorithm~\ref{algo:deepdlalgo2}]
Let $E_{\ell \rightarrow L}$ denote the event $ \{  \hat{\A}^{(\ell \rightarrow L)} = \A^{(\ell \rightarrow L)}\}$, $\ell = 1,\ldots,L$. Let $1 \leq \bar{\ell} \leq L$, then $\forall \bar{\ell}$
	\begin{equation}
		\mathbb{P}[\cap_{\ell=1}^{\bar{\ell}} E_{\ell \rightarrow L}] 
		 \geq \prod_{\ell=1}^{\bar{\ell}} (1-2\delta_{(\ell \rightarrow L)}).
		\label{eq:exactdictreco2}
	\end{equation}
	\label{thm:mainres2}
\end{thm}

\noindent The proof of the theorem is in the appendix.

If we apply the theorem with $\bar{\ell} = L$, it states that, with the given probability, we can learn all of the product matrices in the deep sparse generative model, and hence all $L$ dictionaries. Assumption A5a is a statement about the complexity of this learning. i.e., the amount of data required: the computational complexity is $\mathcal{O} \left(\text{ max}(r_1^2,r_1 M^2_{(0)} s) \text{log}\left(\frac{2 r_1 }{\delta_{(1 \rightarrow L)}}\right)\right)$. In neural network terminology, $r_1$ is the dimension of the deepest layer's output, i.e., the number of hidden units in that layer. The sparsity $s$ represents the (target) number of active units at that layer. Assumption A5b states that a sufficient condition for the forward-factorization algorithm to succeed is that $r_{\ell} = \mathcal{O} \left(\text{ max}(r_{\ell+1}^2,r_{\ell+1} M^2_{(\ell)} s_{(\ell)}) \text{log}\left(\frac{2 r_{\ell+1} }{\delta_{(\ell+1 \rightarrow L)}}\right)\right)$. In other words, the number of hidden units at given layer ought to be a function of the number of hidden units at the preceding layer.

Together, we can interpret the assumptions of the deep sparse generative model, as well as assumptions A5a and A5b, as \emph{prescriptions} for how to design deep versions of the auto-encoders from Section~\ref{sec:shallow}. We can summarize these prescriptions as follows
\begin{enumerate}
	\item The user first picks the number of hidden units $r_L$ in the first layer. This layer can have more units than the dimension $d_L$ of the data. The matrix of weights connecting the data to the hidden units from the first layer can be dense.
	\item Following the first layer, the network ought to expand. That is, the user ought to progressively increase the number of hidden units at all following layers according to assumption A5b. Moreover, the matrices of weights connecting hidden units from consecutive pairs of layers ought to be sparse.
	\item Finally, according to assumption A5a, the user ought to use training data of size proportional to the number of hidden units in the deepest layer.
\end{enumerate}
%The thorem This can be interpreted as a statement regarding the complexity of learning deep versions of the recurrent auto-encoders from Section~\ref{sec:shallow}. Indeed, in neural networks terminology, $r_\ell$ is the size of the embedding at $\ell^\text{th}$ layer and $s_{\mathbf{Y}^{(\ell-1)}}$ the number active neurons at that layer. 
Similar to~\cite{altmin}, our simulations (Section~\ref{sec:sims}) suggest that quadratic terms in assumptions A5a and A5b are artifacts of the techniques from~\cite{altmin} for proving the convergence of Algorithm~\ref{algo:altminalgo} in the shallow case. Specifically, the second-order terms come from the proof of Lemma A.6 in~\cite{altmin}. That is, the simulations suggest that the learning complexity depends on product of the number of active neurons in the deepest layer and the number of hidden units in that layer. Moreover, after picking the number of hidden units in the first layer, the simulations suggest that the number of hidden units at a given layer ought to increase linearly with both the target number of active units in that layer (sparsity), and the number of hidden units in its input layer. An interesting line of research would be to improve the bounds from Lemma A.6 in~\cite{altmin}.

Assumptions A6 states that the error bound of the sparse-coding step ought to decrease as a function of iterations: this is natural if indeed the estimates of the dictionaries converge to the true dictionaries, so should the error of the sparse coding step, which comes from using the ``wrong'' dictionary. Finally, assumption A7 states how close we should initialize dictionaries as a function of the model parameters.

\subsubsection{Backward-factorization algorithm}

Algorithm~\ref{algo:deepdlalgo} specifies the steps of the backward-factorization algorithm. 

Let us gain some intuition by considering the case of $L=2$. First, we learn $\A^{(2)}$ and $\mathbf{Y}^{(1)}$ using Algorithm~\ref{algo:altminalgo} which, under regularity assumptions, guarantees that, with high probability,  $(\hat{\A}^{(2)},\hat{\mathbf{Y}}^{(1)}) = ({\A}^{(2)},{\mathbf{Y}}^{(1)})$. Then, we solve for $\A^{(1)}$ and $\mathbf{X}$ similarly. Appealing to Theorem 3.1 from~\cite{altmin}, we can conclude that, with high probability, we have solved for $\A^{(2)}$, $\A^{(1)}$ and $\mathbf{X}$. The second step assumes that the sparse matrix $\A^{(1)}$ satisfies RIP. We show this in Section~\ref{sec:sprand}.

\begin{algorithm}
%\DontPrintSemicolon % Some LaTeX compilers require you to use \dontprintsemicolon    instead
\KwIn{Samples $\mathbf{Y}$, number of levels $1 \leq \bar{\ell} \leq L$, initial dictionary estimates $\{\A^{(\ell)}(0)\}_{\ell=\bar{\ell}}^L$, accuracy sequences $\{\epsilon^{(\ell)}_t\}_{\ell=\bar{\ell}}^L$, sparsity levels $\{s_{\mathbf{Y}^{(\ell-1)}}\}_{\ell=\bar{\ell}}^L$, scaling sequence $\{\sigma_{(0 \rightarrow \ell-1)}\}_{\ell=\bar{\ell}}^{L}$. %Thresholding function $\text{T}_\rho(a) = a$ if $|a| > \rho$ and $0$ o.w.
}
$\hat{\mathbf{Y}}^{(L)} = \mathbf{Y}$\\
$\ell = L$\\
\While{$\ell \geq \bar{\ell}$}{
	$\left(\hat{\A}^{(\ell)},\hat{\mathbf{Y}}^{(\ell-1)}\right) = \text{AltMinDict}(\sigma_{(0 \rightarrow \ell-1)}\hat{\mathbf{Y}}^{(\ell)},\A^{(\ell)}(0),\epsilon^{(\ell)}_t,s_{\mathbf{Y}^{(\ell-1)}},\infty)$\\
	$\ell = \ell - 1$
}
\KwOut{$\{\hat{\A}^{(\ell)},\hat{\mathbf{Y}}^{(\ell-1)}\}_{\ell=\bar{\ell}}^{L}$}
\caption{Backward-factorization algorithm for deep dictionary learning}
\label{algo:deepdlalgo}
\end{algorithm}

We now state assumptions on the deep generative model of Equation~\ref{eq:deepdl} guaranteeing that Algorithm~\ref{algo:deepdlalgo} can successfully recover all the dictionaries. 

% These assumptions will let us give guarantees and sample-complexity estimates for the success, for arbitrary $L$, of the sequential alternating-minimization algorithm described above for $L=2$.
%The reader can compare these assumptions to assumptions A1--A7 from~\cite{altmin}. As in~\cite{altmin}, we assume, without any loss in generality that the columns of $\{\A^{(\ell)}\}_{\ell=1}^L$ all have unit $\ell_2$ norm, i.e. $\vectornorm{\avec_j^{(\ell)}}_2 = 1$, $j=1,\cdots,r_\ell$, $\ell = 1,\cdots,L$.

\paragraph{Assumptions:} Let $\mathbf{Y}^{(0)} = \mathbf{X}$, $s_{\mathbf{Y}^{(0)}} = s$, and $\forall \ell=1,\ldots,L$, $s_{\mathbf{Y}^{(\ell)}} = s \prod_{\ell'=1}^\ell s_{(\ell)}$, $\ell = 1,\ldots,L-1$. 

\noindent Further define the scalar sequence $\{\sigma_{(0 \rightarrow \ell)}\}_{\ell=0}^{L-1}$ as follows
\begin{eqnarray}
	\sigma_{(0 \rightarrow 0)} & = & 1, \\
	\sigma_{(0 \rightarrow \ell)} & = & \prod_{\ell'=0}^\ell \sqrt{s_{(\ell')}}.
\end{eqnarray}
\noindent For each $\ell=1,\ldots,L-1$, the purpose of $\sigma_{(0\rightarrow \ell)}$ is to scale the sparse matrix $\Y^{(\ell)}$ so that its nonzero entries have variance centered around $1$ (we detail this in the appendix, where we sketch of the proof of convergence of the algorithm).

\begin{itemize}
	\item[(B1)] \textbf{Dictionary Matrices satisfying RIP}: For each $\ell=1,\ldots,L$, the dictionary matrix $\A^{(\ell)}$ has $2s_{\mathbf{Y}^{(\ell-1)}}$-RIP constant of $\delta_{2s_{\mathbf{Y}^{(\ell-1)}}} < 0.1$.
	\item[(B2a)] \textbf{Spectral Condition of all Dictionaries:} For each $\ell=1,\ldots,L$, the dictionary matrix $\A^{(\ell)}$ has bounded spectral norm, i.e., for some constant $\mu_{(\ell)} > 0$, $\vectornorm{\A^{(\ell)}}_2 < \mu_{(\ell)} \sqrt{\frac{r_\ell}{d_\ell}}$.
	\item[(B2b)] \textbf{Spectral Condition of Sparse Dictionaries:} For each $\ell=1,\ldots,L-1$, the transpose of the dictionary matrix $\A^{(\ell)}$ has bounded smallest singular value, i.e., for some constant $\tilde{\mu}_{(\ell)} > 0$, $\sigma_{\text{min}}({\A^{(\ell)\text{T}}}) > \tilde{\mu}_{(\ell)} \sqrt{\frac{r_\ell}{d_\ell}}$.
	\item[(B2c)] \textbf{Spectral Condition of Indicator Matrices of Sparse Dictionaries:} For each $\ell=1,\ldots,L-1$, the spectral norm of the matrix $\U^{(\ell)}$ of indicator values of the dictionary matrix $\A^{(\ell)}$ satisfies $\vectornorm{\U^{(\ell)}}_2 \leq 2 \sqrt{\frac{s_{(\ell)}^2 r_\ell}{r_{\ell+1}}}$.	
		
	\item[(B3)] \textbf{Non-zero Entries in Coefficient Matrix}: The non-zero entries of $\mathbf{X}$ are drawn i.i.d. from a distribution with \emph{mean zero} such that $\mathbb{E}[(\mathbf{X}_{ij})^2] = 1$, and satisfy the following a.s.: $|\mathbf{X}_{ij}| \leq M_{(0)}, \forall, i,j$.
	\item[(B4)] \textbf{Sparse Coefficient Matrix}: The columns of the coefficient matrix have $s$ non-zero entries which are selected uniformly at random from the set of all $s$-sized subsets of $\{1,\ldots,r_1\}$. It is required that $s \leq \frac{d_1^{1/6}}{c_2 {\mu_{(1)}}^{1/3}}$, for some universal constant $c_2$. We further require that, for $\forall \ell=1,\ldots,L-1$, $s_{\mathbf{Y}^{(\ell)}} \leq \frac{d_{\ell+1}^{1/6}}{c_2 {\mu_{(\ell+1)}}^{1/3}}$.
	\item[(B5)] \textbf{Sample Complexity}: Given failure parameters $\{\delta_\ell\}_{\ell=1}^L > 0$, the number of samples $n$ needs to be
%	\begin{equation}
%	\mathcal{O} \left( \max_\ell \text{max}\left(r_1 r_\ell \frac{s_{\mathbf{Y}^{(\ell-1)}}}{s} \text{log}(\frac{2r_1}{\delta_\ell}),r_\ell M^2_{\Y^{(\ell-1)}} s_{\mathbf{Y}^{(\ell-1)}}\text{log}(\frac{2r_\ell}{\delta_\ell})\right) \right).
%		\label{eq:cmp3}
%	\end{equation}
\begin{eqnarray}
	\mathcal{O}  \Big(\max_\ell & & \!\!\!\!\!\!\!\!\!\!\! \text{ max } \Big( r_1 r_\ell \frac{s_{\mathbf{Y}^{(\ell-1)}}}{s} \text{log}(\frac{2r_1}{\delta_\ell}), \nonumber \\
	& & r_\ell M^2_{\Y^{(\ell-1)}} s_{\mathbf{Y}^{(\ell-1)}}\text{log}(\frac{2r_\ell}{\delta_\ell}) \Big)\Big).
	\label{eq:cmp3}
\end{eqnarray}
	Here $|\sigma_{(0 \rightarrow \ell)} \mathbf{Y}^{(\ell)}_{ij}| \leq M_{\Y^{(\ell-1)}}$, $\ell=1,\ldots,L-1$, and $|\mathbf{Y}^{(0)}_{ij}| = |\mathbf{X}_{ij}|  \leq M_{(0)}$.
	\item[(B6)] \textbf{Initial Dictionary with Guaranteed Error Bound}: It is assumed that, $\forall \ell = 1,\ldots,L$, we have access to an initial dictionary estimate $\A^{(\ell)}$ such that 
	\begin{equation}
	\max_{i \in \{1,\ldots,r_\ell\}}\min_{z \in \{-1,1\}} \vectornorm{ z \A_i^{(\ell)}(0)-\A_i^{(\ell)} }_2 \leq \frac{1}{2592 s_{\mathbf{Y}^{(\ell-1)}}^2}.
	\end{equation}
	\item[(B7)] \textbf{Choice of Parameters for Alternating Minimization}: For all $\ell = 1,\ldots,L$, line 4 of Algorithm~\ref{algo:deepdlalgo} uses a sequence of accuracy parameters $\epsilon^{(\ell)}_0 = \mathcal{O}\left(\frac{1}{2592s_{\mathbf{Y}^{(\ell-1)}}^2}\right)$ and
	\begin{equation}
	\epsilon^{(\ell)}_{t+1} = \frac{25050\mu_\ell s_{\mathbf{Y}^{(\ell-1)}}^3}{\sqrt{d_\ell}}\epsilon^{(\ell)}_t.
	\end{equation}
\end{itemize}

We are now in a position to state a theorem on the ability of the backward-factorization algorithm to learn the deep generative model of Equation~\ref{eq:deepdl}, i.e., recover $\{\A^{(\ell)}\}_{\ell=1}^L$ under assumptions B1--B7.

%\subsection{Learning the ``deep'' sparse coding model by sequential alternating minimization}
%
%Algorithm~\ref{algo:deepdlalgo} describes the ``deep'' dictionary learning algorithm. The Algorithm requires the specification of a variable $1 \leq \bar{\ell} \leq L$. Given $1 \leq \bar{\ell} \leq L$, Algorithm~\ref{algo:deepdlalgo} solves for $\{\hat{\A}^{(\ell)}\}_{\ell=\bar{\ell}}^{L}$, starting with $\hat{\A}^{(L)}$ and then sequentially working its to $\hat{\A}^{(\bar{\ell})}$.

\begin{thm}[Exact recovery of the deep generative model by Algorithm~\ref{algo:deepdlalgo}]
Let us denote by $E_\ell$ the event $ \{  \hat{\A}^{(\ell)} = \A^{(\ell)}\}$, $\ell = 1,\ldots,L$. Let $1 \leq \bar{\ell} \leq L$, then $\forall \bar{\ell}$
	\begin{equation}
		\mathbb{P}[\cap_{\ell=\bar{\ell}}^L E_\ell] \geq \prod_{\ell=\bar{\ell}}^L (1-2\delta_\ell).
		\label{eq:exactdictreco}
	\end{equation}
	\label{thm:mainres}
\end{thm}

The Theorem states that, with the given probability, we can learn all of the $L$ transformations in the deep sparse generative model. Random matrices satisfy assumptions B1, B2a, B2b with high probability. This is well known for dense matrices. Section~\ref{sec:sprand} shows this for a class of sparse random matrices with fixed column sparsity. Assumption B2c also holds with high probability for these matrices (Lemma A.1 from~\cite{altmin}). Compared to the proof of Theorem~\ref{thm:mainres2}, the proof of Theorem~\ref{thm:mainres} is less immediate. We give a sketch of the proof in the appendix.

\noindent \underline{\textbf{{Comparing the complexities of Algorithms~\ref{algo:deepdlalgo2} and~\ref{algo:deepdlalgo}}}}:In Equation~\ref{eq:cmp3}, the term corresponding to $\ell=1$ is the same as Equation~\ref{eq:cmpfwd}. Therefore, in the absence of the maximum over layers in Equation~\ref{eq:cmp3}, Algorithms~\ref{algo:deepdlalgo2} and Algorithm~\ref{algo:deepdlalgo} would have the same complexity. Because of the maximum over layers, however, Algorithm~\ref{algo:deepdlalgo} has complexity bigger than Algorithm~\ref{algo:deepdlalgo2}, i.e., it requires more data. Intuitively, this is because assumptions B1--B7 puts fewer constraints on the deep generative model than assumptions A1--A7 do. Indeed, assumption A5b is a stronger constraint on the relationship between the number of hidden units in successive layers than assumption B2b which only requires that, for all $\ell=1,\ldots,L-1$, $r_\ell \geq r_{\ell+1}$. In other words, compared to Theorem~\ref{thm:mainres2}, Theorem~\ref{thm:mainres} applies to a larger set of deep generative models.

\noindent \underline{\textbf{Some limitations of Theorem~\ref{thm:mainres2} and Theorem~\ref{thm:mainres}}}: One limitation of the theorems is that some of the assumptions require the matrices from Equation~\ref{eq:deepdl} to have very large dimensions. In other words, for small (order $1000$) values of the dimensions, some of the assumptions cannot necessarily be met. In assumption A1, for instance, the larger $L$, the more stringent the conditions on RIP constants of the individual dictionaries that would guarantee that the product dictionaries satisfy the assumption. This also applies to assumption B1, in which the sparsity of $\y^{(\ell)}$ decreases with increasing $\ell$ and $L$, making it harder to satisfy.  Another limitation is assumption B4 on $s_{\Y^{(\ell)}}$, which cannot be satisfied easily for small dimensions.  Assumptions A7 and B7, which require the initial dictionaries to be close to the true ones are also limiting. However, close initialization is a standard assumption in convergence analyses of dictionary-learning algorithms. Finally, Theorem~\ref{thm:mainres2} assumes that the sparse dictionaries are random, which may be limiting. Similar to~\cite{parker2014bilinear1,parker2014bilinear2}, which takes a Bayesian approach to dictionary learning, we can think of this assumption as a prior on the sparse dictionaries.%As stated, the theorems ought to be interpreted in the regime in which the dimensions of the matrices involved are large. 

\noindent \underline{\textbf{Conjecture on a lower bound on complexity}}: Our analysis of Algorithms~\ref{algo:deepdlalgo2} and~\ref{algo:deepdlalgo} gives upper bounds on the number of examples sufficient to learn the dictionaries from the deep sparse coding model. Here, we use a so-called ``coupon-collection'' argument to conjecture on a lower bound. The coupon-collector's problem~\cite{anceaume2015new} is a classical one that asks the question: given $r$ coupons, in expectation, how many coupons $n$ does one need to draw without replacement in order to have drawn each coupon at least once? In the context of the shallow sparse coding model (Equation~\ref{eq:shallowdl}, Equation~\ref{eq:deepdl} 
$L=1$), the $r$ columns of the dictionary play the role of the coupons and the number of draws $n$ gives a lower bound on the number
of examples required for learning: for successful dictionary learning, we need to observe each column of the dictionary at least ones. In~\cite{spielman2012exact}, the authors use a generalization of the coupon-collector's problem~\cite{anceaume2015new}, in which one selects $s$ coupons without replacement at every draw, to argue that a lower-bound for the number of examples required for dictionary learning is  $\mathcal{O}(\frac{r}{s} \text{log }r)$. The authors proceed to prove this lower bound for the case of a \emph{square} dictionary. To our knowledge, there is no similar analysis for the case of an overcomplete dictionary and noiseless observations. While there exist analyses for the noisy case based on information-theoretic arguments~\cite{jung2016minimax}, the accompanying results are not not meaningful in the noiseless case for which the SNR is infinite.
 
Using a coupon-collection argument, let us start with a lower bound for the backward-factorization algorithm (Algorithm~\ref{algo:deepdlalgo}). In the end, the lower bound ought to be independent of the sequence in which we perform the factorization and, in fact, we argue that it is. The lower bound for this algorithm is $ \mathcal{O}(\max_\ell  \frac{r_\ell}{s_{Y^{(\ell-1)}}} \text{log } r_\ell)$. Since $r_1 \geq r_2 \cdots \geq r_L$, and $s = s_{\Y^{(0)}} \leq  s_{\Y^{(1)}} \leq \cdots s_{\Y^{(L-1)}}$, the lower bound $\mathcal{O}(\frac{r_1}{s} \text{log }r_1)$. For the forward-factorization algorithm (Algorithm~\ref{algo:deepdlalgo2}), the first step of the algorithm, which yields the product matrix, dictates the sample complexity $n$. The coupon-collection argument yields a lower bound $\mathcal{O}(\frac{r_1}{s} \text{log }r_1)$.

It is not hard to see that the sequence in which one obtains the dictionary is not important for the lower bound. This is because, ultimately, any factorization must obtain, as part of its steps, the matrix of sparse codes $\X$. The coupon-collection argument suggests a lower bound of $\mathcal{O}(\frac{r_1}{s} \text{log }r_1)$ to obtain this matrix. The effect of the sequence in which one factors the production dictionary simply results in constraining the relationship between the number of hidden units at different layers.

\section{Concentration of eigenvalues of column-sparse random matrices with dependent sub-Gaussian entries}
\label{sec:sprand}

The proof of our main results, Theorems~\ref{thm:mainres2} and~\ref{thm:mainres}, rely on sparse matrices satisfying RIP~\cite{candes2008restricted}. In this section, we show that a class of random sparse matrices indeed satisfies RIP.

\subsection{Sparse random sub-Gaussian matrix model}

\noindent Let $\A \in \R^{d \times r}$ be a matrix with $r$ columns $(\avec_j)_{j=1}^r$. Let $\U \in \R^{d \times r}$ be a binary random matrix with columns $(\mathbf{u}_j)_{j=1}^r \in \{0,1\}^d$ that are i.i.d. $s_\A$-sparse binary random vectors each obtained by selecting $s_\A$ entries from $\mathbf{u}_i$ \emph{without} replacement, and letting $\U_{ij}$ be the indicator random variable of whether a given entry $j=1,\ldots,d$ was selected. Let $\mathbf{V} \in \R^{d \times r}$ be a random matrix with i.i.d. entries distributed according to a zero-mean sub-Gaussian random variable $V$ with variance $1$, $|V| \leq 1$ almost surely, and sub-Gaussian norm $\vectornorm{V}_{\psi_2}$--we adopt the notation $\vectornorm{\cdot}_{\psi_2}$ from~\cite{vershynin2010introduction} to denote the sub-Gaussian norm of a random variable. We consider the following generative model for the entries of $\A$:
\begin{equation}
	\A_{ij} = \sqrt{\frac{d}{s_\A}} \U_{ij} \vmat_{ij}, i=1,\ldots,d;\text{ }j=1,\ldots,r.
	\label{eq:sparseA}
\end{equation}

\noindent It is not hard to verify that the random matrix $\A$ thus obtained is such that $E[\avec_i \avec_i^\text{T}] = \eye$. To see this, we note the following properties of the generative model for $\A$
\begin{itemize}
	\item $\mathbb{P}[\U_{ij} = 1] = \frac{s_\A}{d}$.
	\item Let $j \neq j'$, $\mathbb{P}[\U_{ij} = 1 \cap \U_{ij'} = 1] = \frac{s_\A}{d} \cdot \frac{s_\A-1}{d-1}$.
	\item $\mathbb{E}[\A_{ij}^2] = \mathbb{E}\left[\left(\sqrt{\frac{d}{s_\A}} \U_{ij} \vmat_{ij}\right)^2\right] = \frac{d}{s_\A} \cdot \mathbb{P}[\U_{ij} = 1] \cdot \mathbb{E}[\vmat_{ij}^2] = 1$.
	\item Let $j \neq j'$, $\mathbb{E}[\A_{ij}\A_{ij'}] = \mathbb{E}[\U_{ij}\U_{ij'}\vmat_{ij}\vmat_{ij'}] = \mathbb{E}[\U_{ij}\U_{ij'}]\underbrace{\mathbb{E}[\vmat_{ij}\vmat_{ij'}]}_{0} = 0$.
	\item $\vectornorm{\avec_j}_2^2 \leq s_\A \cdot \frac{d}{s_\A} = d$ a.s, $j=1,\ldots,r$. 
\end{itemize}

\noindent Ultimately, we would like to understand the concentration behavior of the singular values of 1) $\A^\text{T}$, and 2) sub-matrices of $\A$ that consist of a sparse subset of columns (RIP-like results). We fist recall the following result from non-asymptotic random matrix theory~\cite{vershynin2010introduction}, and apply it obtain a concentration result on the singular values of the matrix $\A^\text{T}$.  

\begin{thm}[Restatement of Theorem 5.39 from~\cite{vershynin2010introduction} (Sub-Gaussian rows)]
	Let $\W \in \R^{r \times d}$ matrix whose rows $\{(\W^\text{T})_j\}_{j=1}^r$ ($(\W^\text{T})_j$ is the $j^\text{th}$ column of $\W^\text{T}$) are independent sub-Gaussian isotropic random vectors in $\R^d$. Then for every $t \geq 0$, with probability at least $1 - \text{exp}(-ct^2)$ one has
	\begin{equation}
		\sqrt{r} - C\sqrt{d} - t \leq \sigma_{\text{min}}(\W) \leq \sigma_{\text{max}}(\W) \leq \sqrt{r} + C\sqrt{d} + t.
	\end{equation}
\noindent Here, $C = C_K$, $c = c_K \geq 0$ depend only on the sub-Gaussian norm $K = \text{max}_j \vectornorm{(\W^\text{T})_j}_{\psi_2}$ of the rows.
	\label{thm:versh539}
\end{thm}

Before we can apply the above result to $\W = \A^\text{T}$, we need to demonstrate that the columns of $\A$ are sub-Gaussian random vectors, defined as follows

\begin{defn}[Definition 5.22 from~\cite{vershynin2010introduction} (Sub-Gaussian random vectors)]
	We say that a random vector $\x$ in $\R^d$ is sub-Gaussian if the one-dimensional marginals $\ip{\x}{\z}$ are sub-Gaussian random variables for all $\z$ in $\R^d$. The sub-Gaussian norm of $\x$ is defined as
	\begin{equation}
		\vectornorm{\x}_{\psi_2} = \text{sup}_{\z \in \mathcal{S}^{d-1}} \vectornorm{\ip{\x}{\z}}_{\psi_2}.
	\end{equation}
	\label{def:subG}
\end{defn}

%\noindent We first demonstrate that the entries of $A_{i,j}$, i.e. projections onto the standard basis in $\R^n$, are sub-Gaussian random variables. The reason why it's not necessarily obvious that $(\avec_i)_{i=1}^N$ are sub-Gaussian random vectors is because the entries of $\A$ are uncorrelated, but not independent, so that we can't simply apply results on weighted sums of independent sub-Gaussian random variables to conclude that $\A$ satisfies Definition~\ref{def:subG}.

%\begin{thm}[Entries of $\A$ are sub-Gaussian]
%For every $i=1,\cdots,N$, $j=1,\cdots,n$,
%	\begin{equation}
%		\vectornorm{A_{i,j}}_{\psi_2} = \frac{s_\A}{n} \vectornorm{V}_{\psi_2}.
%	\end{equation}
%	\label{thm:Aijsubn}
%\end{thm}
%
%\begin{proof}
%	This follows from the definition of the sub-Gaussian norm:
%	\begin{equation}
%		\vectornorm{A_{i,j}}_{\psi_2} = \text{sup}_{p \geq 1} \mathbb{E}[|A_{i,j}|^p] = \text{sup}_{p \geq 1} \mathbb{E}[U_{i,j}|V_{i,j}|^p] = \frac{s_\A}{n} \text{sup}_{p \geq 1} \mathbb{E}[|V_{i,j}|^p] = \frac{s_\A}{n}  \vectornorm{V}_{\psi_2}.
%	\end{equation}
%\end{proof}

\begin{thm}[The columns of $\A$ are sub-Gaussian random vectors]
For every $j=1,\ldots,r$, $\avec_j$ is a sub-Gaussian random vector. Moreover,
	\begin{equation}
		\vectornorm{\avec_j}_{\psi_2}^2 \leq \frac{d}{s_\A}C_1 \vectornorm{V}_{\psi_2}^2,
	\end{equation}
where $C_1$ is a universal constant.
	\label{thm:aisubn}
\end{thm}

\noindent The proof of the theorem is in the appendix.

\subsection{Concentration properties of minimum and maximum singular values of $\A$}

\noindent We now have all of the requisites to apply Theorem~\ref{thm:versh539} to the matrix $\A$ defined in Equation~\ref{eq:sparseA}.

\begin{lem}
	Let $\A$ be the sparse random matrix obtained according to Equation~\ref{eq:sparseA}. There exist universal constants $c$ and $C$ such that with probability at least $1 - \text{exp}\left(-c t^2) \right)$
	\begin{equation}
		\sqrt{r}(1-\delta) \leq \sigma_{\text{min}}(\A^\text{T}) \leq \sigma_{\text{max}}(\A^\text{T}) \leq \sqrt{r}(1+\delta)
	\end{equation}
	and
	\begin{equation}
		\vectornorm{\frac{1}{r} \A\A^{\text{T}} - \eye}_2 \leq \text{max}(\delta,\delta^2), \delta = C\sqrt{\frac{d}{r}} + \frac{t}{\sqrt{r}}.
 	\end{equation}
	\label{lem:sparseconc}
\end{lem}

\begin{proof}
	The first part of the Lemma follows from applying the Theorem~\ref{thm:versh539} to $\A^\text{T}$ from Equation~\ref{eq:sparseA}. The second part follows from Lemma 5.36 in~\cite{vershynin2010introduction} showing that the  equivalence of the first and second parts. 
\end{proof}

%\noindent \underline{\textbf{Remark 8}}: According to Theorem~\ref{thm:versh539}, the constants $c$ and $C$ depend only on the sub-Gaussian norm of the columns of $\A$, which Theorem~\ref{thm:aisubn} gives a bound for. This bound highlights the main difference between the case when the $\A$ is a sparse matrix according to our model, as opposed to a dense matrix. Indeed, when $\A$ is dense, $s_\A = d$, so that the bound from Theorem~\ref{thm:aisubn} reduces to known bounds on the sub-Gaussian norm of a matrix with i.i.d. sub-Gaussian entries. When $\A$ is sparse, Theorem~\ref{thm:aisubn} implies that the constant $C$ is larger than in the dense case, leading to looser bounds on the minimum and maximum eigenvalues of th $\A^\text{T}$. The constant $c$, which is inversely proportional to the sub-Gaussian norm of the columns of $\A$, is smaller than in dense case, leading to a smaller probability that the bounds from the theorem hold.

\subsection{Near-isometry properties of subsets of columns of $\A$}

Consider a subset $T_0 \subset \{1,2,\ldots,r\}$ s.t. $|T_0| = s$, and let $\A_{T_0}$ denote the $d \times s$ matrix corresponding to the subset of columns of $\A$ from Equation~\ref{eq:sparseA} indexed by $T_0$. We want an RIP-like result of $\A_{T_0}$, i.e., we seek a bound for the difference between the spectral norm of  $\A_{T_0}^{\text{T}}\A_{T_0}$ (appropriately normalized) and the identity matrix.
\begin{defn}[Restatement of Definition from~\cite{vershynin2010introduction} (Restricted isometries)]
A $d \times r$ matrix $\A$ satisfies the restricted isometry property of order $s \geq 1$ if there exists $\delta_s \geq 0$ such that the inequality
	\begin{equation}
		(1-\delta_s) \vectornorm{\x}_2^2 \leq \vectornorm{\A\x}_2^2 \leq (1+\delta_s) \vectornorm{\x}_2^2
	\end{equation}
	\label{def:ripdef}
holds for all $\x \in \R^r$ with $|\text{Supp}(\x)| \leq s$. The smallest number $\delta_s = \delta_s(\A)$ is called the restricted isometry constant of $\A$.
\end{defn}

We first recall a result from~\cite{vershynin2010introduction} which, along with the results from the previous section, will give us the desired bound. The result applies to one of two models for $\A$.

\noindent \underline{\textbf{Row-independent model}}: the rows of $\A$ are independent sub-Gaussian isotropic random vectors in $\R^r$.

\noindent \underline{\textbf{Column-independent model}}: The columns $\avec_j$ of $\A$ are independent sub-Gaussian isotropic random vectors in $\R^d$ with $\vectornorm{\avec_j}_2 = \sqrt{d}$ a.s.

\noindent The column-independent model is the one of interest here.

\begin{thm}[Restatement of Theorem 5.65 from~\cite{vershynin2010introduction} (Sub-Gaussian restricted isometries)]
	Let $\A \in \R^{d \times r}$ sub-Gaussian random matrix with independent rows or columns, which follows either of the two models above. Then the normalized matrix $\bar{\A} = \frac{1}{\sqrt{d}} \A$ satisfies the following for every sparsity level $1 \leq s \leq r$ and every number $\delta \in (0,1)$:
	\begin{equation}
		\text{if } d \geq C \delta^{-2} s \text{log} (e r/s), \text{ then } \delta_s(\bar{\A}) \leq \delta
	\end{equation}
\noindent with probability at least $1-2\text{exp}(-c\delta^2d)$. Here, $C = C_K$, $c = c_K \geq 0$ depend only on the sub-Gaussian norm $K = \text{max}_j \vectornorm{\avec_j}_{\psi_2}$ of the rows or columns of $\A$.
	\label{thm:versh565}
\end{thm}

We can apply this theorem to $\A$ from Equation~\ref{eq:sparseA} to conclude that, with a sufficient number of measurements $d$, $\bar{\A}$ satisfies the RIP of order $s$.

\begin{lem}
	Let $\A$ be the sparse random matrix obtained according to Equation~\ref{eq:sparseA}. Then the normalized matrix $\bar{\A} = \frac{1}{\sqrt{d}} \A$ satisfies the following for every sparsity level $1 \leq s \leq r$ and every number $\delta \in (0,1)$:
	\begin{equation}
		\text{if } d \geq C \delta^{-2} s \text{log} (e r/s), \text{ then } \delta_s(\bar{\A}) \leq \delta
	\end{equation}
\noindent with probability at least $1-2\text{exp}(-c\delta^2d)$, where $C = C_K$, $c = c_K \geq 0$ depend only on the bound from Equation~\ref{eq:aisubn3} on the sub-Gaussian norm $\vectornorm{\avec_j}_{\psi_2}$ of the columns of $\A$.
	\label{lem:sparsrip}
\end{lem}

\begin{proof}
	By construction, $\A$ follows the column-independent model. In addition, we have proved in Theorem~\ref{thm:aisubn} that the columns of $\A$ are sub-Gaussian random vectors. The result follows from applying Theorem~\ref{thm:versh565} to $\A$ from Equation~\ref{eq:sparseA}.
\end{proof}

\noindent \underline{\textbf{Remark 4}}: Note that $\bar{\A}_{ij} = \frac{1}{\sqrt{d}} \frac{\sqrt{d}}{\sqrt{s_\A}} \U_{ij}\vmat_{ij} = \frac{1}{\sqrt{s_\A}} \U_{ij} \vmat_{ij} $.

\section{Simulations}
\label{sec:sims}

We use simulated data to demonstrate the ability of Algorithms~\ref{algo:deepdlalgo2} and~\ref{algo:deepdlalgo} to learn the deep generative model from Equation~\ref{eq:deepdl}, and to assess their sensitivity to the distance between the initial dictionaries and true ones. In addition, we compare the algorithms to an auto-encoder architecture designed specifically to learn the deep generative model. To be clear, the simulations are not meant to be exhaustive. Their goal is to demonstrate some aspects of the theoretical results.

\noindent \paragraph{Simulations.} As both algorithms are computationally demanding, in that they require the solutions to many convex optimization problems, we consider the case when $L=2$. We used the simulation studies from~\cite{altmin} to guide our choice of parameters
\begin{enumerate}
	\item We chose $\A^{(2)}$ to be of size $100 \times 200$, i.e., $d_2 = 100$ and $r_2 = 200$. We chose its entries to be i.i.d. $\mathcal{N}(0,\frac{1}{d_2})$.
	\item We chose $\A^{(1)}$ to be of size $200 \times 800$, i.e., $d_1 = 200$ and $r_1 = 800$. We chose its entries according to the sparse random matrix model from Equation~\ref{eq:sparseA}, letting $\{\vmat_{ij}\}_{i=1,j=1}^{200,800}$ be i.i.d. Rademacher random variables ($+/-1$ with equal probability). Each column of $\A^{(1)}$ has sparsity level $s_{(1)} = 3$.
	\item We chose $n = 6400$, so that $\X$ is of size $800 \times 6400$. Similar to~\cite{altmin}, we chose the non-zero entries of $\X$ to be i.i.d $\text{U}([-2,-1]\cup[1,2]$), and picked $s = 3$ (sparsity level of each column of $\X$).
	
	\item For a given matrix $\A$ and its estimate $\hat{\A}$, we use the following error metric to assess the success of the algorithm at learning $\A$
	\begin{equation}
	\text{err}(\hat{\A},\A) = \text{max}_i \sqrt{1 - \frac{\ip{\avec_i}{\hat{\avec}_i}^2}{\|{\avec}_i\|_2^2\|\hat{\avec}_i\|_2^2}}.
	\end{equation}
\end{enumerate}
\begin{figure}[H]
  \centering
  \includegraphics[scale=0.6]{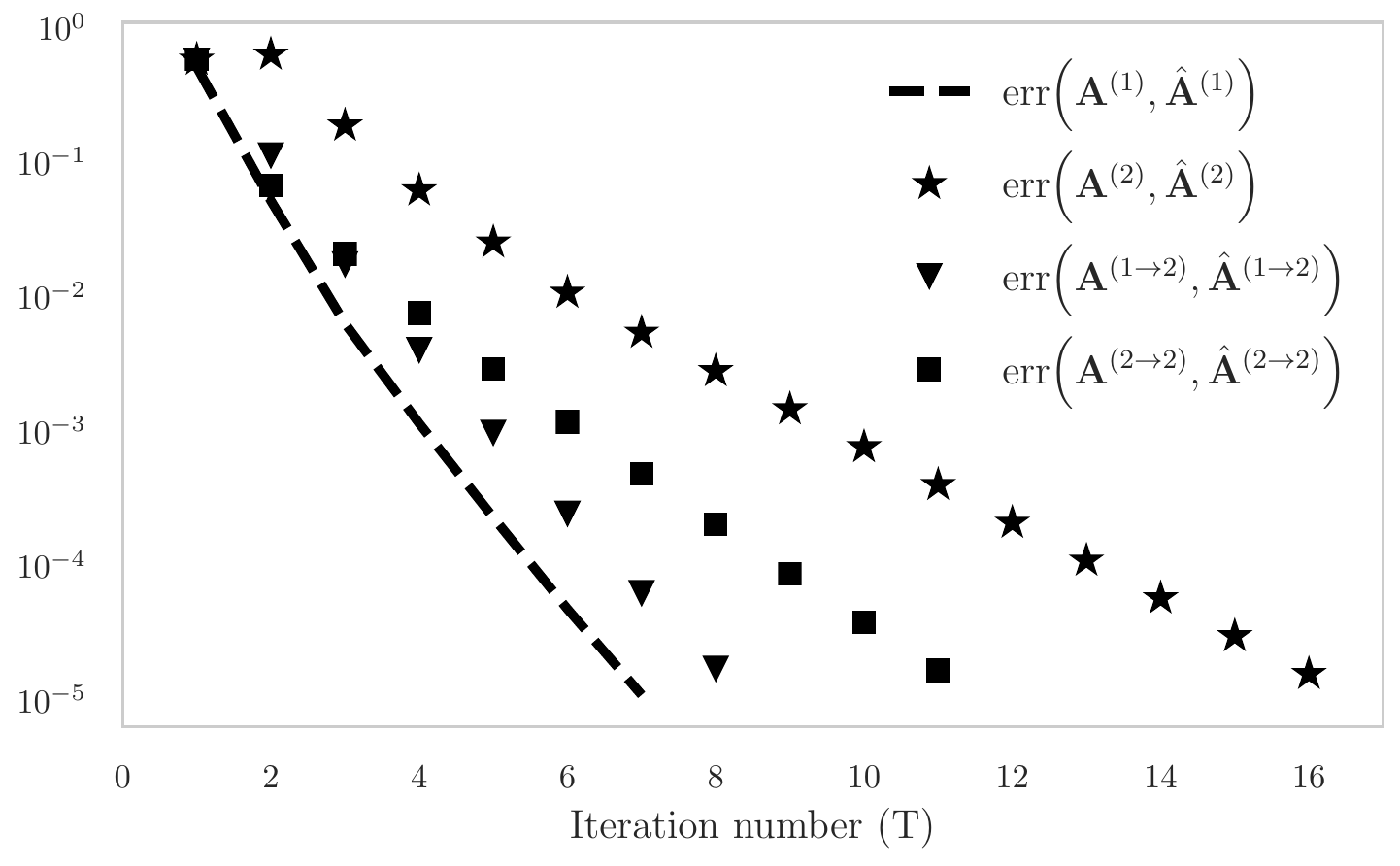}
\caption{Semi-log plot of the error between the true dictionaries and the recovered ones as a function of the number of iterations of Algorithm~\ref{algo:altminalgo}. The figure shows that the recovered dictionaries converge to the true ones linearly (exponentially). We refer the reader to the main text for additional interpretations of the simulations. }
\label{fig:recovery_errors}
\end{figure}

\noindent \underline{\textbf{N.B.}}: Note the fact that $n$ is much smaller than $r_1^2$ and $r_2^2$. Yet, the simulations will demonstrate that the dictionaries can be learned exactly. As in~\cite{altmin}, this highlights the fact that the second-order terms in the complexity bounds is an artifact of the proof techniques from~\cite{altmin}, which form the basis for our own proofs.

\noindent \paragraph{Initialization.} To initialize Algorithms~\ref{algo:deepdlalgo} and~\ref{algo:deepdlalgo2}, we need to initialize $\A^{(1)}$, $\A^{(2)}$ and $\A^{(2 \rightarrow 2)}$. For a generic matrix $\A$ of size $d \times r$,  following the simulation studies from~\cite{altmin}, we let $\A(0) = \A + \mathbf{Z}$, where entries of $\mathbf{Z}$ are i.i.d. $ \alpha \cdot\mathcal{N}(0,\frac{1}{d})$, $\alpha \in (0,1)$. We define the signal-to-noise ratio (SNR) of $\A(0)$ as follows
\begin{equation}
	\text{SNR} = 10 \text{ log}_{10} \left( \frac{\sigma^2_\mathbf{A}}{\sigma^2_\mathbf{Z}} \right) = -10 \text{ log}_{10} \alpha^2.
\end{equation}

\noindent \paragraph{Implementation.} We implemented Algorithms~\ref{algo:deepdlalgo2} and~\ref{algo:deepdlalgo} in the Python programming language. We used {\tt cvxpy}~\cite{cvxpy} with the {\tt MOSEK}~\cite{mosek} solver to find the solutions to the optimization problems from the inner loop of the alternating-minimization procedure (Algorithm~\ref{algo:altminalgo}). The authors in~\cite{altmin} use the GradeS~\cite{garg2009gradient} algorithm for this inner loop because it is faster. In our experience, cvxpy was much more stable numerically. The inner loop of Algorithm~\ref{algo:altminalgo} is embarrassingly parallelizable. Therefore, we also implemented a distributed version of our algorithm using the Python module {\tt dask}~\cite{dask}. The code is hosted on {\tt bitbucket}. The author is happy to make it available upon request\footnote{\url{https://bitbucket.org/demba/ds2p/src/master/}}. 

\noindent \paragraph{Results.} We assessed the ability of Algorithms~\ref{algo:deepdlalgo2} and~\ref{algo:deepdlalgo} to recover their respective dictionaries, as well as the sensitivity of the algorithms to the SNR of the initial dictionaries.

\noindent \underline{\textbf{Dictionary learning at moderate SNR}}: We applied the algorithms to data simulated as described previously. We generated the initial dictionaries so that their SNR equals $6$ dB, which corresponds to $\alpha = 0.5$. Figure~\ref{fig:recovery_errors} depicts the error between the true dictionaries and the ones recovered using Algorithm~\ref{algo:deepdlalgo2} and Algorithm~\ref{algo:deepdlalgo}. We obtained $\hat{\A}^{(1 \rightarrow 2)}$ and $\hat{\A}^{(2 \rightarrow 2)}$ using Algorithm~\ref{algo:deepdlalgo}, while $\hat{\A}^{(1)}$ and $\hat{\A}^{(2)}$ were obtained using Algorithm~\ref{algo:deepdlalgo}. %For Algorithm~\ref{algo:deepdlalgo2}, we chose to plot the recovery error for $\hat{\A}^{(1 \rightarrow 2)}$ to demonstrate its ability to recover the product matrix. 
The difference in the number of iterations comes from our choice of termination criterion at a value on the order of $10^{-5}$.

\begin{figure}[H]
  \centering
  \includegraphics[scale=0.6]{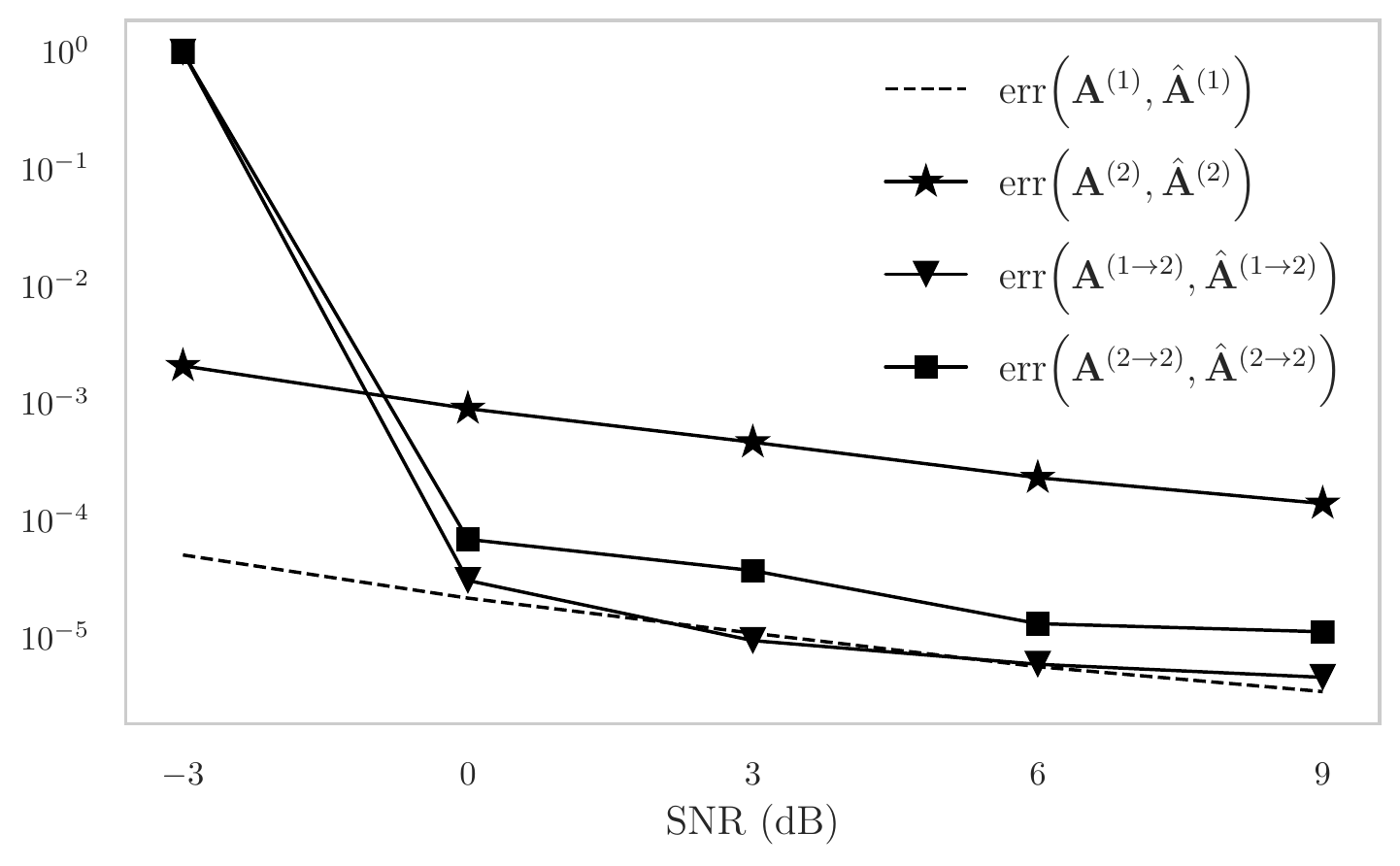}
\caption{Semi-log plot of the error between the true dictionaries and the recovered ones as a function of the SNR of the initial dictionaries. The figure shows that both algorithms are robust to initialization errors for moderate to high SNRs. We refer the reader to the main text for additional interpretations of the simulations. }
\label{fig:snr_sensitivity}
\end{figure}

The figure demonstrates the linear (exponential) convergence of the recovered dictionary to the true ones. The error curve for $\hat{\A}^{(2)}$ is consistent with the simulation studies from~\cite{altmin}. In this example ($L=2$), $\A^{(2 \rightarrow 2)}$ and $\A^{(2)}$ are the same matrix. The error rate appears to be faster for the former compared to the latter. We hypothesize that this is because the sparsity-level parameter for $\A^{(2 \rightarrow 2)}$ is $s_{(1)}=3$, while that for $\A^{(2)}$ is $s\cdot s_{(1)} = 9$. The amount of data required in Theorem 3.1 from~\cite{altmin} scales linearly with the sparsity level. The error rate for $\A^{(1)}$ is much faster than all of the others, i.e., it would appear that $\A^{(1)}$ is easier to learn. This should not be surprising as $\A^{(1)}$ is a matrix that is very sparse ($3$-sparse in fact). Therefore, its columns have far fewer that $d_1 = 100$ degrees of freedom. The proof techniques utilized in~\cite{altmin}, and which we rely upon, do not explicitly take into account the sparsity of $\A^{(1)}$. All that the theorem gives is an upper bound on the rate of convergence, without an accompanying lower bound. Studying the effect of sparsity of the matrix of interest in dictionary learning would be an interesting area of further inquiry. Overall, our simulations demonstrate the ability of Algorithms~\ref{algo:deepdlalgo2} and~\ref{algo:deepdlalgo} to recover the deep generative model from Equation~\ref{eq:deepdl}. We would like to stress that Figure~\ref{fig:recovery_errors} required on the order of $5$ hours to generate on a PC with $8$ GB of RAM and $4$ cores with $2$ threads each. 

\noindent \underline{\textbf{Sensitivity to SNR of initialization}}: We repeated the simulations for $5$ equally-spaced values of the SNR of the initial dictionaries, beginning at $-3$ dB and ending at $9$ dB. Figure~\ref{fig:snr_sensitivity} plots the error between the true dictionaries and the ones recovered as a function of SNR. We ran all algorithms for $T = 10$ iterations. The figure demonstrates that both Algorithm~\ref{algo:deepdlalgo2} and~\ref{algo:deepdlalgo} perform well at moderate to high SNRs. Consistent with the findings from Figure~\ref{fig:recovery_errors}, the forward-factorization algorithm converges faster overall and does better at high SNRs. The backward factorization algorithm appears to be more robust at low SNRs. To estimate $\A^{(2)}$, the backward-factorization algorithm uses $n = 6400$ examples, while the forward-factorization algorithm only uses $r_1 = 800$ examples. This would explain why the former performs better at low SNRs.

\noindent While, in practice, the simulations may or may not satisfy all assumptions from the theorems, they provide a useful demonstration of the algorithms.

\begin{figure}[H]
  \centering
  \includegraphics[scale=0.9]{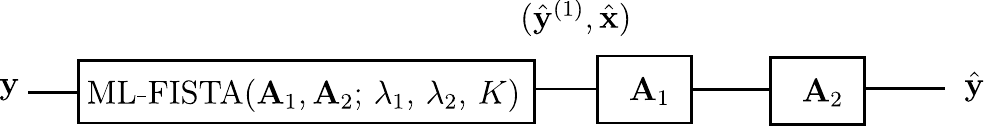}
\caption{Auto-encoder architecture for learning the dictionaries in the deep-sparse coding model with $L=2$. The encoder performs $K$ iterations of the ML-FISTA algorithm. The decoder reconstructs the observations by applying the dictionaries to the output of the encoder.}
\label{fig:ml-fista_ae}
\end{figure}

\noindent \underline{\textbf{Comparison with an auto-encoder with shared weights}}: Recent work~\cite{nguyen2019dynamics} suggests that a two-layer auto-encoder, whose encoder performs one iteration of ISTA~\cite{beck2009fast}, and whose decoder is linear and share weights with the encoder (Figure~\ref{fig:ista_ae}, $K = 1$), can learn the dictionary in the shallow sparse-coding model approximately. It is reasonable, then, to assume that a multi-layer generalization of the architecture from ought to be able to the dictionaries from the deep-sparse coding model approximately. Figure~\ref{fig:ml-fista_ae} proposes such an architecture. The encoder performs multi-layer FISTA~\cite{sulam2019multi} (ML-FISTA). ML-FISTA is a generalization of the multi-layer ISTA~\cite{sulam2019multi} (ML-ISTA) algorithm to solve
\begin{equation}
	\min\limits_{\x\in \R^{r_1}} \frac{1}{2} \vectornorm{\y-\A_2\A_1\x}_2^2 + \lambda_2 \overbrace{\vectornorm{\A_1\x}_1}^{\vectornorm{\y^{(1)}}_1} + \lambda_1 \vectornorm{\x}_1 .
	\label{eq:ml-lasso}
\end{equation}
\noindent ML-ISTA generalizes the ISTA algorithm for the shallow sparse coding to an algorithm for deep sparse coding. ML-FISTA is a fast version of ML-ISTA. Similar to the interpretation of ISTA as a recurrent architecture, we can	 interpret ML-FISTA as a recurrent network with weights $\A_1$ and $\A_2$ and soft-thresholding (two-sided ReLU nonlinearities). The encoder of Figure~\ref{fig:ml-fista_ae} unfolds $K$ iterations of ML-FISTA. The decoder is linear and shares weights with the encoder. It first applies $\A_1$ and then $\A_2$ to the output of the decoder. We train the auto-encoder by mini-batch stochastic gradient descent and the ADAM optimizer~\cite{Kingma2014AdamAM} with a learning rate of $0.1$. We used $K=1000$, $\lambda_1 = $ and $\lambda_2 = 0.01$. We performed a grid search and found these values to be the ones that produce, respectively, sparse codes $\hat{\x}$ and $\hat{\y}^{(1)}$ that approximate the ground-truth codes from the simulations the best. We trained for $150$ epochs and used a batch size of $40$ examples. Similar to ISTA (Figure~\ref{fig:ista_ae}), ML-FISTA uses parameters $M_1$ and $M_2$ for its ReLUs, both of which we set equal to $10$. Figure~\ref{fig:ml_fista_errors} shows the result of training the ML-FISTA auto-encoder. The SNR of the initial dictionary is $6$ dB. The figure shows that training the architecture decreases the error from its initial value of $\approx 0.5$ to values of $\approx 0.2$ and $\approx 0.1$ for $\A_1$ and $\A_2$ respectively. These values are orders of magnitude higher than those achieved by Algorithms~\ref{algo:deepdlalgo2} and~\ref{algo:deepdlalgo}. These results are consistent with the experiments from~\cite{nguyen2019dynamics}. To our knowledge, there is no analysis in the literature of auto-encoder architectures for learning dictionaries from a deep sparse coding model, nor are there experiments assessing the ability of various architectures to learn the dictionaries from such models. We hypothesize that the reason why the errors plateau as a function of epoch is because the values of $\lambda_1$ and $\lambda_2$ are constant. These two parameters play the role of the $\epsilon_t$ sequences in Algorithms~\ref{algo:deepdlalgo2} and~\ref{algo:deepdlalgo} and, therefore, ought to decrease, both as a function of ML-FISTA unfoldings~\cite{chen2018theoretical}, and as a function of epoch~\cite{altmin}. Unfortunately, a theoretical characterization of how to do this is absent from the literature at this time.

\begin{figure}[H]
  \centering
  \includegraphics[scale=0.6]{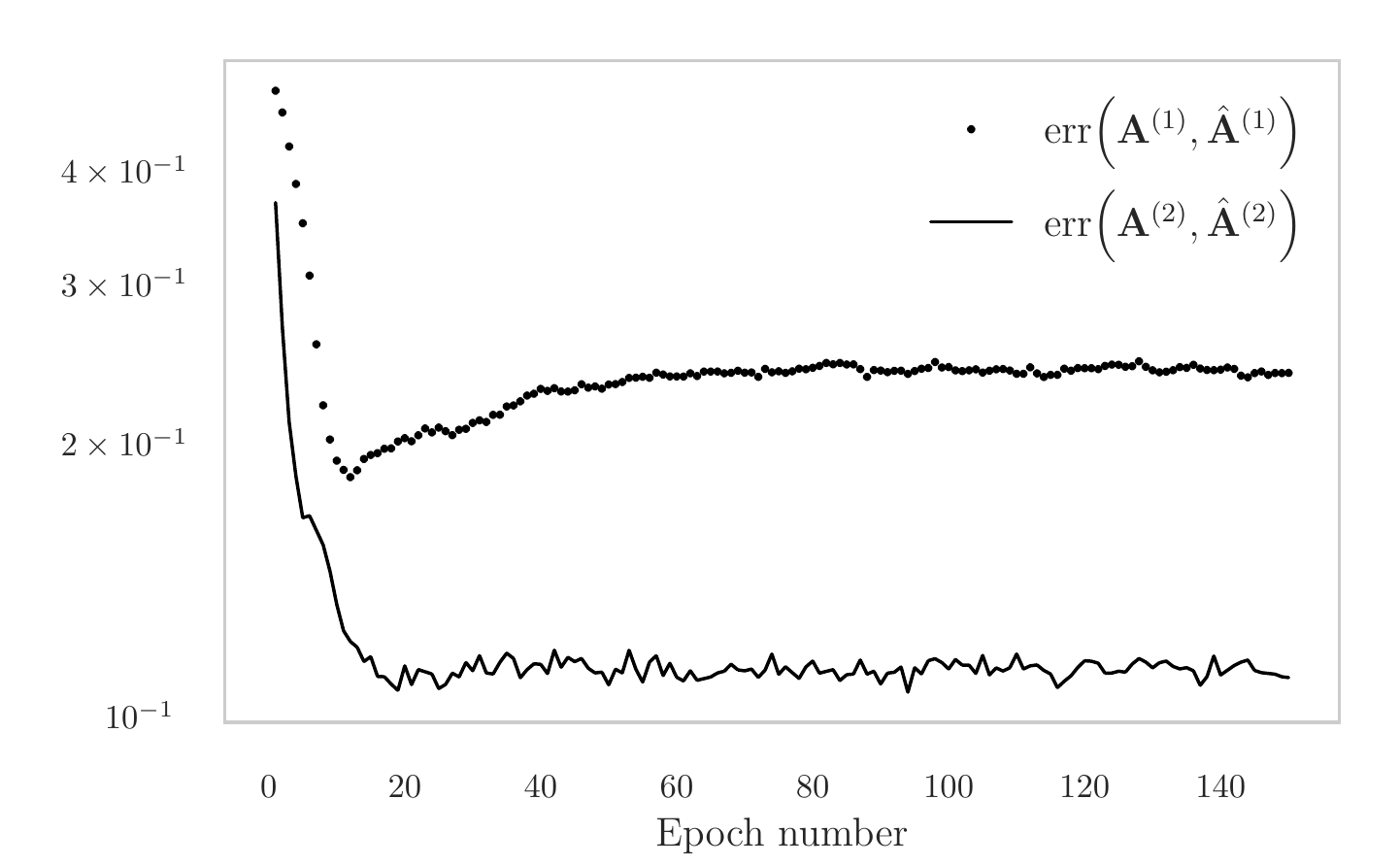}
\caption{Semi-log plot of the error between the true dictionaries and the  ones recovered by the ML-FISTA auto-encoder as a function of epoch.}
\label{fig:ml_fista_errors}
\end{figure}

%\section{Sparse Signal Representations}
%
%\subsection{Sparse coding model}
%
%\subsection{Dictionary and Sparse Code Estimation}
%
%\subsection{Approximation by shallow auto-encoder}
%
%\section{Deep Sparse Signal Representations}
%%
%%\subsection{Synthesis DWT Filter Banks with Sparsity Constraints}
%%
%%\subsection{Intuitive Parallels with Deep Neural Networks}
%%
%%\section{Exact and Stable Recovery of Deep Sparse Signal Representations}
%%
%%
%%\subsection{Exact Recovery}
%%
%%
%%\subsection{Stability}
%
%\section{Deep sparse coding model}
%
%\subsection{Estimation by sequential alternate minimization}
%
%\subsection{Theoretical Analysis of Sequential Alternate Minimization}
%
%\subsection{Approximation by deep auto-encoder}

\section{Discussion}
\label{sec:disco}

We have provided insights as to the complexity of learning deep ReLU neural networks by building a link between deep recurrent auto-encoders~\cite{rolfe2013discriminative} and classical dictionary learning theory~\cite{altmin}. We used this link to develop a deep version of the classical sparse coding model from dictionary learning. Starting with a sparse code, a cascade of linear transformations are applied in succession to generate an observation. Each transformation in the cascade is constrained to have sparse columns, except for the last. We developed two sequential alternating-minimization algorithms to learn this deep generative model from observations and proved that, under assumptions stated in detail, the algorithms are able to learn the underlying transformations in the cascade that generated the observations. The computational complexity of the algorithms is a function of the dimension of the inputs of each of the transformations in the cascade (number of hidden units), the sparsity level of these inputs (number of active neurons), and the sparsity of the columns of the sparse matrices (sparsity of weight matrix in sparse layers).

In particular, the computational complexity of the forward-factorization is $\mathcal{O} \left(\text{ max}(r_1^2,r_1 s_{\Y^{(0)}} )\right)$, i.e., a function of the number of hidden units in the deepest layer and the number active neurons in that layer. In addition, the forward-factorization algorithm requires $r_{\ell} = \mathcal{O} \left( \text{ max}(r_{\ell+1}^2,r_{\ell+1} s_{(\ell)}) \right)$, $\ell=1,\ldots,L-1$, i.e., it relates the number of hidden units at a given layer to that at the preceding layer, giving a practical prescription for designing deep ReLU architectures. The complexity of the backward-factorization algorithm is $\mathcal{O} \left(\text{max}_{\ell}  \text{ max}(r_1^2,r_\ell s_{\mathbf{Y}^{(\ell-1)}})\right) $: it requires more data because it puts less stringent constraints on the architecture. The proofs rely on a certain family of sparse matrices satisfying the RIP. We studied the properties of random versions of this family of matrices using results from non-asymptotic random matrix theory. We used a coupon-collection~\cite{anceaume2015new} argument to conjecture a lower bound on sample complexity. This argument posits that we $\mathcal{O}(\frac{r_1}{s} \text{log }r_1)$ for learning all of the matrices in the deep sparse coding model, where $r_1$ is the number of hidden units at the deepest layer, and $s$ the number of active neurons. This lower bound may explain why deep neural networks require more data to train. In the deep sparse coding model, each layer gives a sparse representation of the input. The deeper the layer, the sparser, and  hence the more 	specific, the representation. Consider a scenario in which  we wish to use the  sparse representations at different depths for a classification task. The deeper the representation, the sparser it is, and hence the easier the classification task, as it would require fewer features. %Deep networks outperform shallow ones on classification tasks. The argument that increased sparsity helps 	classification, along with the lower bound argument, give a clue as to why deeper 	representations require more data to learn. 
We require $\mathcal{O}(\frac{r_\ell}{s_{\Y{(\ell-1)}} }\text{log }r_\ell)$ examples to learn the all the weights up to layer $\ell$ and, hence, the sparse representations at layer $\ell$. The deeper we go, the sparser the representation and, therefore, the more data we need. Whether a deep network outperforms a shallower one may depend on whether it can result in a sparser, more specific, representation that is useful for a classification task.

We used simulations to validate the algorithms we propose for learning the dictionaries in the deep sparse coding model. The simulations show that the proposed algorithms can successfully learn a sequence of dictionaries and that they are robust to initialization at a wide range of SNRs. We also showed that the algorithms are superior to an auto-encoder architecture motivated by the multi-layer FISTA algorithm. %Similar to~\cite{altmin}, the simulations (Section~\ref{sec:sims}) suggest that the second-order terms in our sample complexity bounds are artifacts of the proof techniques from~\cite{altmin}, which we rely on to arrive at our main result.

%In future work, we will study the case when the matrices in the cascade satisfy the mutual-incoherence property rather than RIP. In dictionary learning, good initialization procedures are known for the mutually-incoherent case~\cite{altmin}. As they stand, our results do not provide guarantees for a case when we know good initialization procedures. It will be interesting to study the coherence properties of sparse random matrices. As the deep dictionary learning algorithm relies on the solution of convex optimization problems, we will explore distributed implementations of these algorithms, as well as ones on GPUs. One impressive fact about deep learning is the GPU-based infrastructure that has been developed to train virtually any kind of network imaginable. It would be of great interest to develop a similar infrastructure%, backed by TensorFlow~\cite{abadi2016tensorflow} 
%for sparsity-regularized inverse problems. %\textcolor{red}{We also plan to train shallow and deep versions of the recurrent sparse auto-encoders to demonstrate that they can solve the shallow and deep dictionary learning problems respectively. We have already demonstrated this in~\cite{tolooshams2018} for the convolutional case}. 

\section*{Acknowledgments}
I would like to thank Ms. Bahareh Tolooshams for her assistance with the neural-network experiments. I would also like to thank the anonymous referees for their feedback.

\section{Appendix}

\subsection{Derivation of soft-thresholding operator}

The solution $\hat{x}$ must satisfy
\begin{equation}
	y = x + \lambda \partial |x|.
	\label{eq:softopt}
\end{equation}

\noindent \underline{\textbf{Case $y > \lambda$}}: The idea is to show that, if $y > \lambda$, then necessarily $x > 0$. Suppose for a contradiction that $x < 0$, then Equation~\ref{eq:softopt} yields $y = x - \lambda < 0$, a contradiction. Similarly, if $x = 0$, then $y = \lambda \partial x$, where the sub-gradient $\partial x \in [-1,1]$, a contradiction since $y > \lambda$. Therefore, if $y > \lambda$, Equation~\ref{eq:softopt} yields $\hat{x} = y - \lambda$.

\noindent \underline{\textbf{Case $y < \lambda$}}: This case proceeds similarly as above.

\noindent \underline{\textbf{Case $|y| \leq \lambda$}}: Suppose, for contradiction, that $x > 0$, then Equation~\ref{eq:softopt} yields $y > \lambda$. Similarly, if $x < 0$, we obtain $y < -\lambda$. Therefore, if $|y| \leq \lambda$, $\hat{x} = 0$.

\noindent Together, the three cases above yield the soft-thresholding function of Equation~\ref{eq:softop} as the solution to Equation~\ref{eq:uncsoft}.

\subsection{Proof of Theorem~\ref{thm:prodrip}}

We proceed by induction on $\bar{\ell}$.\\

\begin{proof}
\noindent \underline{Base case}: $\bar{\ell} = L$. The theorem is true for this case by assumption A1.
\\
\\
\noindent \underline{Induction}: Suppose the theorem is true for $\bar{\ell}$, we will show that it holds true for $\bar{\ell}-1$. Let $\y^{(\bar{\ell}-2)}$ be a $2s_{\mathbf{Y}^{(\bar{\ell}-2)}}$-sparse vector
\begin{equation}
	\vectornorm{\A^{(\bar{\ell}-1 \rightarrow L)}\y^{(\bar{\ell}-2)}}_2^2 = \vectornorm{\A^{(\bar{\ell} \rightarrow L)}\A^{(\bar{\ell}-1)}\y^{(\bar{\ell}-2)}}_2^2.
\end{equation}

$\A^{(\bar{\ell}-1)}\y^{(\bar{\ell}-2)}$ is a $2s_{\mathbf{Y}^{(\bar{\ell}-1)}}$-sparse vector, allowing us to apply our inductive hypothesis
\begin{eqnarray}
	& & \prod_{\ell=\bar{\ell}}^L (1-\delta_{2s_{\mathbf{Y}^{(\ell-1)}}}) \vectornorm{\A^{(\bar{\ell}-1)}\y^{(\bar{\ell}-2)}}_2^2 \leq \nonumber \\
    & & \vectornorm{\A^{(\bar{\ell} \rightarrow L)}\A^{(\bar{\ell}-1)}\y^{(\bar{\ell}-2)}}_2^2 \leq \nonumber \\
    & & \prod_{\ell=\bar{\ell}}^L (1 + \delta_{2s_{\mathbf{Y}^{(\ell-1)}}}) \vectornorm{\A^{(\bar{\ell}-1)}\y^{(\bar{\ell}-2)}}_2^2.
\end{eqnarray}

The result follows by assumption A1 since $\A^{(\bar{\ell}-1)}$ satisfies the RIP of order $2s_{\mathbf{Y}^{(\bar{\ell}-2)}}$.

\end{proof}

A direct consequence of the theorem is that $\forall \bar{\ell} = 1,\ldots,L$, the RIP constant of $\A^{(\bar{\ell} \rightarrow L)}$ must be smaller than or equal to $\text{max}\left(1-\prod_{\ell=\bar{\ell}}^L (1-\delta_{2s_{\mathbf{Y}^{(\ell-1)}}}),\prod_{\ell=\bar{\ell}}^L (1+\delta_{2s_{\mathbf{Y}^{(\ell-1)}}})-1\right)$.

\subsection{Proof of Theorem~\ref{thm:mainres2}}

\begin{proof}
We proceed by induction on $\bar{\ell}$.\\
\noindent \underline{Base case}: $\bar{\ell} = 1$. In this case, $\mathbf{Y} = \A^{(1 \rightarrow L)} \X$. Theorem 3.1 from~\cite{altmin} guarantees that $\hat{\A}^{(1 \rightarrow L)}$, the limit as $T \to \infty$ of $\A^{(1 \rightarrow L)}(T)$ converges to $\A^{(1 \rightarrow L)}$ with probability at least $1-2\delta_{(1 \rightarrow L)}$. Therefore, $\mathbb{P}[E_{1 \rightarrow L}] \geq 1-2\delta_{1 \rightarrow L}$, proving the base case.
\\
\\
\noindent \underline{Induction}: Suppose the Theorem is true for $\bar{\ell}$, we will show that is true for $\bar{\ell} + 1$.

\noindent Conditioned on the event $\cap_{\ell=1}^{\bar{\ell}} E_{\ell \rightarrow L}$, 
\begin{equation}
 \hat{\A}^{(\bar{\ell} \rightarrow L)} = \A^{(\bar{\ell} \rightarrow L)} = \A^{(\bar{\ell}+1 \rightarrow L)}\A^{(\bar{\ell})}.
\end{equation}
\noindent The nonzero entries of $\A^{(\bar{\ell})}$ do not have variance $1$, which is one of the assumptions from Theorem 3.1 of~\cite{altmin}. That's why, we scale $\hat{\A}^{(\bar{\ell} \rightarrow L)}$ in line 5 of Algorithm~\ref{algo:deepdlalgo2} by $\sqrt{s_{\bar{\ell}}}$
\begin{eqnarray}
\sqrt{s_{\bar{\ell}}}\hat{\A}^{(\bar{\ell} \rightarrow L)} & = & \sqrt{s_{\bar{\ell}}}\A^{(\bar{\ell} \rightarrow L)} \\
 & = & \sqrt{s_{\bar{\ell}}}\A^{(\bar{\ell}+1 \rightarrow L)}\A^{(\bar{\ell})} \\ 
 & = & \A^{(\bar{\ell}+1 \rightarrow L)} \left(\sqrt{s_{\bar{\ell}}}\A^{(\bar{\ell})} \right).
\end{eqnarray}

The nonzero entries of $\sqrt{s_{\bar{\ell}}}\A^{(\bar{\ell})}$ have variance $1$. Therefore, applying Theorem 3.1 from~\cite{altmin}, the limit $\hat{\A}^{(\bar{\ell}+1 \rightarrow L)}$ as $T \to \infty$ of $\A^{(\bar{\ell}+1 \rightarrow L)}(T)$ converges to $\A^{(\bar{\ell} + 1 \rightarrow L)}$ with probability at least $1-2\delta_{(\bar{\ell}+1 \rightarrow L)}$. Hence, $\mathbb{P}[\cap_{\ell=1}^{\bar{\ell}+1} E_{\ell \rightarrow L}] = \mathbb{P}[E_{\bar{\ell}+1 \rightarrow L} | \cap_{\ell=1}^{\bar{\ell}} E_{\ell \rightarrow L} ] \mathbb{P}[\cap_{\ell=1}^{\bar{\ell}} E_{\ell \rightarrow L}] = (1-2\delta_{(\bar{\ell}+1 \rightarrow L)}) \prod_{\ell=1}^{\bar{\ell}} (1-2\delta_{(\ell \rightarrow L)}) = \prod_{\ell=1}^{\bar{\ell}+1} (1-2\delta_{(\ell \rightarrow L)})$.

\end{proof}

\noindent This completes the proof.

\subsection{Sketch of proof of Theorem~\ref{thm:mainres}}

We will prove the result by induction on $\bar{\ell}$. Compared to Theorem~\ref{thm:mainres2}, the proof of Theorem~\ref{thm:mainres} is less immediate. Indeed, the the proof of Theorem~\ref{thm:mainres2} follows from application of Theorem 3.1 from~\cite{altmin} and induction, assuming the product matrices satisfy RIP (Section~\ref{sec:sprand}). The proof of Theorem~\ref{thm:mainres} requires us to fill in some steps for proofs of the main results from~\cite{altmin} that lead to Theorem 3.1. We sketch how to fill in these steps below. Before proceeding with the proof, let us discuss in detail the case when $L=2$ in Equation~\ref{eq:deepdl} and $\bar{\ell}=1$. We focus on exact recovery of $\A^{(2)}$ and $\A^{(1)}$ and defer computation of the probability in Equation~\ref{eq:exactdictreco} to the proof that will follow.

\noindent \underline{\textbf{{Intuition behind the proof:}}} Algorithm~\ref{algo:deepdlalgo} begins by solving for $\left(\hat{\A}^{(2)}, \hat{\mathbf{Y}}^{(1)}\right)$. If we can show that the algorithm succeeds for this pair, in particular that $\hat{\A}^{(2)} = \A^{(2)}$, then it follows that $\hat{\A}^{(1)} = \A^{(1)}$ in the following iteration of the algorithm. This is because, if the first iteration were to succeed, then  $\mathbf{Y}^{(1)} = \A^{(1)} \mathbf{X}$, which is the very model of Equation~\ref{eq:shallowdl}, which was treated in detail~\cite{altmin}. %If we can show that the \emph{sparse} matrix $\A^{(1)}$ follows RIP--topic of the the next section Section~\ref{sec:sprand}--then we can apply Theorem 1 from~\cite{altmin} to guarantee recovery of $\A^{(1)}$.

Focusing on $\mathbf{Y} = \A^{(2)} \mathbf{Y}^{(1)}$, the key point we would like to explain is that, in Equation~\ref{eq:shallowdl}, the properties of $\mathbf{X}$ that allow the authors from~\cite{altmin} to prove their main result also apply to $ \sigma_{(0 \rightarrow 1)} \mathbf{Y} = \sqrt{s_{(1)}} \A^{(2)} \mathbf{Y}^{(1)} = \A^{(2)} (\sqrt{s_{(1)}}\mathbf{Y}^{(1)}) $. This is not directly obvious because  $ \sqrt{s_{(1)}} \mathbf{Y}^{(1)}$ is the product of a sparse matrix and the matrix $\mathbf{X}$ of codes.

We begin with a few remarks on the properties of  $ \sqrt{s_{(1)}} \mathbf{Y}^{(1)}$.

\noindent \underline{\textbf{Deriving a version of Lemma 3.2 from~\cite{altmin}}}: Lemma 3.2 relies on Lemmas A.1 and A.2, which give bounds for the matrix of indicator values for the nonzero entries of $\X$. For $\sqrt{s_{(1)}}\Y^{(1)} = \sqrt{s_{(1)}}\A^{(1)} \X$, we can replace, in the proof of Lemma 3.2, the matrix of indicators of its nonzero entries with the product of the matrices of indicators of the nonzero entries for $\A^{(1)}$ and $\X$ respectively. Using assumption B2c, this yields a bound that now depends on the $n$, $r_2$ and the sparsity level $s \cdot s_{(1)}$.

\noindent \underline{\textbf{Bounded singular values}}: Lemmas A.1 and A.2 from~\cite{altmin} give bounds on the largest and smallest singular values of $\X^{\text{T}}$. Together with assumptions B2a and B2b, this yields bounds for the largest and the smallest singular values of $ \sqrt{s_{(1)}} \mathbf{Y}^{(1)\text{T}}$

\begin{eqnarray}
	\sigma_{\text{min}}(\sqrt{s_{(1)}} \mathbf{Y}^{(1)\text{T}}) & \geq & \tilde{\mu}_{(1)} \sqrt{\frac{n s_{(1)}s}{4 r_2}}\\
	\vectornorm{\sqrt{s_{(1)}} \mathbf{Y}^{(1)\text{T}}}_2 & \leq & \mu_{(1)} 2 \sqrt{\frac{n (s_{(1)}s)^2}{r_2}},
\end{eqnarray}
\noindent where both $\tilde{\mu}_{(1)}$ and $\mu_{(1)}$ are $\mathcal{O}(1)$~\cite{vershynin2010introduction}. A version of Lemmas A.3 and A.4 for $ \sqrt{s_{(1)}} \mathbf{Y}^{(1)}$ directly follows.

\noindent \underline{\textbf{Uniformly bounded variance}}: Since the columns of $\mathbf{X}$ are i.i.d., so are those of $ \mathbf{Y}^{(1)}$. Moreover, since both the entries of $\A^{(1)}$ and $\mathbf{X}$ are bounded by assumptions, so are those of $\mathbf{Y}^{(1)}$. Let us compute the variance of $ \sqrt{s_{(1)}} \mathbf{Y}_{ij}^{(1)}$. First note that
\begin{equation}
  \sqrt{s_{(1)}} \mathbf{Y}_{ij}^{(1)} = \sqrt{s_{(1)}}\mathbf{e}_i^{\text{T}} \mathbf{y}_j^{(1)} = \sqrt{s_{(1)}}\mathbf{e}_i^{\text{T}} \A^{(1)} \x_j,
\end{equation}

\noindent  where $\mathbf{e}_i$ is the $i^{\text{th}}$ element of the standard basis in $\in \R^{r_2}$. Since the nonzero entries of $\X_j$ have mean zero by assumption B3, the entries of $\X_j$ are uncorrelated. This is because the expected value of the product of such entries is the expected value of the product of dependent Bernoulli random variables and independent \emph{mean zero} random variables. Therefore, the variance of  $ \sqrt{s_{(1)}} \mathbf{Y}_{ij}^{(1)}$ is
\begin{equation}
	s_{(1)} \mathbb{E}[(\mathbf{e}_i^{\text{T}} \A^{(1)} \x_j)^2] =  \vectornorm{\A^{(1)\text{T}}\mathbf{e}_i}_2^2 \frac{s_{(1)}s}{r_1},
\end{equation}
\noindent where we have used the fact that, since the nonzero entries of $\X_j$ have mean zero, they are uncorrelated, and we have substituted the variance $\frac{s}{r_1}$ of the entries of $\X_j$. Using assumptions B2a and B2b gives uniform bounds from above and below, respectively, on the variance of $ \sqrt{s_{(1)}} \mathbf{Y}_{ij}^{(1)}$
\begin{equation}
	\tilde{\mu}_{(1)}^2 \frac{s_{(1)}s}{r_2} \leq \mathbb{E}[(\sqrt{s_{(1)}} \mathbf{Y}_{ij}^{(1)})^2] \leq \mu_{(1)}^2 \frac{s_{(1)}s}{r_2},
\end{equation}

\noindent A direct consequence of this are versions of Lemma A.5 and Lemma A.6 that apply to $ \sqrt{s_{(1)}} \mathbf{Y}^{(1)}$.
\\
\\
\noindent The above discussion gives all the necessary ingredients for a version of Lemma 3.3, the center piece of~\cite{altmin}, applied to $\sqrt{s_{(1)}}\Y^{(1)}$. % them follows. 
We can now to apply Theorem 3.1 from~\cite{altmin} to $\sqrt{s_{(1)}}\mathbf{Y} = \A^{(2)} (\sqrt{s_{(1)}} \mathbf{Y}^{(1)})$, guaranteeing recovery of $\A^{(2)}$.

The interested reader can verify all of the above. A detailed technical exposition of these points would lead to a tedious and unnecessary digression, without adding much intuition. Using induction, it can be shown that the above remarks apply to $\sigma_{(0\rightarrow \ell-1)}\Y^{(\ell)} = \A^{(\ell)} (\sigma_{(0\rightarrow \ell-1)} \Y^{(\ell-1)}) $ for all $\ell = 1,\ldots,L-1$.

\begin{proof}
We proceed by induction on $\bar{\ell}$.\\
\noindent \underline{Base case}: $\bar{\ell} = L$. In this case, $\sigma_{(0\rightarrow L-1)}\mathbf{Y} = \A^{(L)} (\sigma_{(0\rightarrow L-1)}\mathbf{Y}^{(L-1)})$. Following the remark above, $\sigma_{(0\rightarrow L-1)}\mathbf{Y}^{(L-1)}$ obeys the properties of $\mathbf{X}$ from Theorem 3.1 in~\cite{altmin}. This theorem guarantees that $\hat{\A}^{(L)}$, the limit as $T \to \infty$ of $\A^{(L)}(T)$ converges to $\A^{(L)}$ with probability at least $1-2\delta_L$. Therefore, $\mathbb{P}[E_L] \geq 1-2\delta_L$, proving the base case.
\\
\\
\noindent \underline{Induction}: Suppose the Theorem is true for $\bar{\ell}$, we will show that is true for $\bar{\ell}-1$.

Conditioned on the event $\cap_{\ell=\bar{\ell}}^L E_\ell$,  $\hat{\Y}^{(\bar{\ell}-1)} = \Y^{(\bar{\ell}-1)}$ and $ \sigma_{(0 \rightarrow \bar{\ell}-2)}\Y^{(\bar{\ell}-1)} = \A^{(\bar{\ell}-1)}(\sigma_{(0 \rightarrow \bar{\ell}-2)}\Y^{(\bar{\ell}-2)})$. Therefore, applying Theorem 3.1 from~\cite{altmin}, the limit $\hat{\A}^{(\bar{\ell}-1)}$ as $T \to \infty$ of $\A^{(\bar{\ell}-1)}(T)$ converges to $\A^{(\bar{\ell}-1)}$ with probability at least $1-2\delta_{\bar{\ell}-1}$. Therefore, $\mathbb{P}[\cap_{\ell=\bar{\ell}-1}^L E_\ell] = \mathbb{P}[E_{\bar{\ell}-1} | \cap_{\ell=\bar{\ell}}^L E_\ell ] \mathbb{P}[\cap_{\ell=\bar{\ell}}^L E_\ell] = (1-2\delta_{\bar{\ell}-1}) \prod_{\ell=\bar{\ell}}^L (1-2\delta_\ell) = \prod_{\ell=\bar{\ell}-1}^L (1-2\delta_\ell)$.
\end{proof}

\noindent This completes the proof.

\subsection{Proof of Theorem~\ref{thm:aisubn}}

\begin{proof}
	We show this by bounding $\vectornorm{\sqrt{\frac{s_\A}{d}}\avec_j}_{\psi_2}$:
{\begin{small} {	
	\begin{eqnarray}
	& & \vectornorm{\sqrt{\frac{s_\A}{d}}\avec_j}_{\psi_2} = \sup_{\z \in \mathcal{S}^{d-1}} \vectornorm{\ip{\sqrt{\frac{s_\A}{d}}\avec_j}{\z}}_{\psi_2}\\
			 & = & \hspace*{-0.22in}	\sup_{\z \in \mathcal{S}^{d-1}} \sup_{p \geq 1} \frac{\mathbb{E}\left[\left|\sum_{i=1}^d z_i \U_{ij} \vmat_{ij}\right|^p\right]^{1/p}}{\sqrt{p}}\\
		& = & \hspace*{-0.22in}	\sup_{\z \in \mathcal{S}^{d-1}} \sup_{p \geq 1} \frac{\left(\mathbb{E}_{(\U_{ij})_{i=1}^d}\mathbb{E}\left[\left|\sum_{i=1}^d z_i \U_{ij} \vmat_{ij}\right|^p \Bigg|  (\U_{ij})_{i=1}^d \right]\right)^{1/p}}{\sqrt{p}}
	\end{eqnarray}
	\begin{eqnarray}
		& \leq & \hspace*{-0.22in}	\sup_{\z \in \mathcal{S}^{d-1}} \sup_{p \geq 1} \frac{\mathbb{E}_{(\U_{ij})_{i=1}^d} \mathbb{E}\left[\left|\sum_{i=1}^d z_i \U_{ij} \vmat_{ij}\right|^p \Bigg|  (\U_{ij})_{i=1}^d \right]^{1/p}}{\sqrt{p}} \\
		& \leq & \hspace*{-0.22in}	\sup_{\z \in \mathcal{S}^{d-1}} \mathbb{E}_{(\U_{ij})_{i=1}^d} \sup_{p \geq 1} \frac{\mathbb{E} \left[\left|\sum_{i=1}^d z_i \U_{ij} \vmat_{ij}\right|^p \Bigg|  (\U_{ij})_{i=1}^d \right]^{1/p}}{\sqrt{p}}.
		\label{eq:aisubn1}
	\end{eqnarray}
	}
\end{small}}

\noindent Conditioned on  $(\U_{ij})_{i=1}^d$, $\sup_{p \geq 1} \frac{\mathbb{E}\left[\left|\sum_{i=1}^d z_i \U_{ij} \vmat_{ij}\right|^p\right]^{1/p}}{\sqrt{p}}$ is the sub-Gaussian norm of the sum of $s_\A$ independent sub-Gaussian random variables. Therefore, according to Lemma 5.9 in~\cite{vershynin2010introduction},
\begin{small}{
	\begin{eqnarray}
	 & &\left(\sup_{p \geq 1} \frac{\mathbb{E}\left[\left|\sum_{i=1}^d z_i \U_{ij} \vmat_{ij}\right|^p \Bigg|  (\U_{ij})_{i=1}^d \right]^{1/p}}{\sqrt{p}}\right)^2 \\ & = &  \vectornorm{\sum_{i \in \{1,\ldots,s_A\}} z_i \vmat_{ij}}_{\psi_2}^2 \\
		& \leq & C_1 \vectornorm{V}_{\psi_2}^2 \sum_{i \in \{1,\ldots,s_\A\}} z_i^2 \\
		& \leq & C_1 \vectornorm{V}_{\psi_2}^2 \vectornorm{\z}_2^2.
		\label{eq:aisubn2}
	\end{eqnarray}
	}
	\end{small}
\noindent Putting Equation~\ref{eq:aisubn2} back into Equation~\ref{eq:aisubn1} yields
\begin{small}{
	\begin{eqnarray}
		\vectornorm{\avec_j}_{\psi_2} & \leq & \sqrt{\frac{d}{s_\A}} \sup_{\z \in \mathcal{S}^{d-1}} \mathbb{E}_{(\U_{ij})_{i=1}^d}\left[ \sqrt{C_1} \vectornorm{V}_{\psi_2}\vectornorm{\z}_2 \right] \\
		& = & \sqrt{\frac{d}{s_\A}} \sup_{\z \in \mathcal{S}^{d-1}}  \sqrt{C_1} \vectornorm{V}_{\psi_2}\vectornorm{\z}_2 \\
		& \leq & \sqrt{\frac{d}{s_\A}} \sqrt{C_1} \vectornorm{V}_{\psi_2}.
		\label{eq:aisubn3}
	\end{eqnarray}
	}
\end{small}
\end{proof}

\noindent This completes the proof.

\ifCLASSOPTIONcaptionsoff
  \newpage
\fi


\begin{thebibliography}{10}
\providecommand{\url}[1]{#1}
\csname url@samestyle\endcsname
\providecommand{\newblock}{\relax}
\providecommand{\bibinfo}[2]{#2}
\providecommand{\BIBentrySTDinterwordspacing}{\spaceskip=0pt\relax}
\providecommand{\BIBentryALTinterwordstretchfactor}{4}
\providecommand{\BIBentryALTinterwordspacing}{\spaceskip=\fontdimen2\font plus
\BIBentryALTinterwordstretchfactor\fontdimen3\font minus
  \fontdimen4\font\relax}
\providecommand{\BIBforeignlanguage}[2]{{%
\expandafter\ifx\csname l@#1\endcsname\relax
\typeout{** WARNING: IEEEtran.bst: No hyphenation pattern has been}%
\typeout{** loaded for the language `#1'. Using the pattern for}%
\typeout{** the default language instead.}%
\else
\language=\csname l@#1\endcsname
\fi
#2}}
\providecommand{\BIBdecl}{\relax}
\BIBdecl

\bibitem{patel2015probabilistic}
A.~B. Patel, T.~Nguyen, and R.~G. Baraniuk, ``A probabilistic theory of deep
  learning,'' \emph{arXiv preprint arXiv:1504.00641}, 2015.

\bibitem{papyan2017convolutional}
V.~Papyan, Y.~Romano, and M.~Elad, ``Convolutional neural networks analyzed via
  convolutional sparse coding,'' \emph{The Journal of Machine Learning
  Research}, vol.~18, no.~1, pp. 2887--2938, 2017.

\bibitem{papyan2017working}
V.~Papyan, J.~Sulam, and M.~Elad, ``Working locally thinking globally:
  Theoretical guarantees for convolutional sparse coding,'' \emph{IEEE
  Transactions on Signal Processing}, vol.~65, no.~21, pp. 5687--5701, 2017.

\bibitem{sulam2018multilayer}
J.~Sulam, V.~Papyan, Y.~Romano, and M.~Elad, ``Multilayer convolutional sparse
  modeling: Pursuit and dictionary learning,'' \emph{IEEE Transactions on
  Signal Processing}, vol.~66, no.~15, pp. 4090--4104, 2018.

\bibitem{ye2018deep}
J.~C. Ye, Y.~Han, and E.~Cha, ``Deep convolutional framelets: A general deep
  learning framework for inverse problems,'' \emph{SIAM Journal on Imaging
  Sciences}, vol.~11, no.~2, pp. 991--1048, 2018.

\bibitem{tishby2015deep}
N.~Tishby and N.~Zaslavsky, ``Deep learning and the information bottleneck
  principle,'' in \emph{Information Theory Workshop (ITW), 2015 IEEE}.\hskip
  1em plus 0.5em minus 0.4em\relax IEEE, 2015, pp. 1--5.

\bibitem{altmin}
A.~Agarwal, A.~Anandkumar, P.~Jain, and P.~Netrapalli, ``Learning sparsely used
  overcomplete dictionaries via alternating minimization,'' \emph{SIAM Journal
  on Optimization}, vol.~26, no.~4, pp. 2775--2799, 2016.

\bibitem{nguyen2019dynamics}
T.~V. Nguyen, R.~K. Wong, and C.~Hegde, ``On the dynamics of gradient descent
  for autoencoders,'' in \emph{The 22nd International Conference on Artificial
  Intelligence and Statistics}, 2019, pp. 2858--2867.

\bibitem{gregor2010learning}
K.~Gregor and Y.~LeCun, ``Learning fast approximations of sparse coding,'' in
  \emph{Proceedings of the 27th International Conference on Machine Learning
  (ICML-10)}, 2010, pp. 399--406.

\bibitem{rolfe2013discriminative}
J.~T. Rolfe and Y.~LeCun, ``Discriminative recurrent sparse auto-encoders,''
  \emph{arXiv preprint arXiv:1301.3775}, 2013.

\bibitem{tolooshams2018scalable}
B.~Tolooshams, S.~Dey, and D.~Ba, ``Scalable convolutional dictionary learning
  with constrained recurrent sparse auto-encoders,'' in \emph{2018 IEEE 28th
  International Workshop on Machine Learning for Signal Processing
  (MLSP)}.\hskip 1em plus 0.5em minus 0.4em\relax IEEE, 2018, pp. 1--6.

\bibitem{chang2019randnet}
T.~Chang, B.~Tolooshams, and D.~Ba, ``Randnet: deep learning with compressed
  measurements of images,'' in \emph{2019 IEEE 29th International Workshop on
  Machine Learning for Signal Processing (MLSP)}.\hskip 1em plus 0.5em minus
  0.4em\relax IEEE, 2019, pp. 1--6.

\bibitem{anceaume2015new}
E.~Anceaume, Y.~Busnel, and B.~Sericola, ``New results on a generalized coupon
  collector problem using markov chains,'' \emph{Journal of Applied
  Probability}, vol.~52, no.~2, pp. 405--418, 2015.

\bibitem{candes2008restricted}
E.~J. Candes \emph{et~al.}, ``The restricted isometry property and its
  implications for compressed sensing,'' \emph{Comptes rendus mathematique},
  vol. 346, no. 9-10, pp. 589--592, 2008.

\bibitem{vershynin2010introduction}
R.~Vershynin, ``Introduction to the non-asymptotic analysis of random
  matrices,'' \emph{arXiv preprint arXiv:1011.3027}, 2010.

\bibitem{sreter2017learned}
H.~Sreter and R.~Giryes, ``Learned convolutional sparse coding,'' \emph{arXiv
  preprint arXiv:1711.00328}, 2017.

\bibitem{glorot2011deep}
X.~Glorot, A.~Bordes, and Y.~Bengio, ``Deep sparse rectifier neural networks,''
  in \emph{Proceedings of the Fourteenth International Conference on Artificial
  Intelligence and Statistics}, 2011, pp. 315--323.

\bibitem{beck2009fast}
A.~Beck and M.~Teboulle, ``A fast iterative shrinkage-thresholding algorithm
  for linear inverse problems,'' \emph{SIAM journal on imaging sciences},
  vol.~2, no.~1, pp. 183--202, 2009.

\bibitem{TolooshamsBahareh2019deepresidualAE}
B.~Tolooshams, S.~Dey, and D.~Ba, ``Deep residual auto-encoders for expectation
  maximization-inspired dictionary learning,'' pp. 1--13, 2019,
  arXiv:1904.08827.

\bibitem{theodoridis2015machine}
S.~Theodoridis, \emph{Machine learning: a Bayesian and optimization
  perspective}.\hskip 1em plus 0.5em minus 0.4em\relax Academic Press, 2015.

\bibitem{aberdam2019multi}
A.~Aberdam, J.~Sulam, and M.~Elad, ``Multi-layer sparse coding: The holistic
  way,'' \emph{SIAM Journal on Mathematics of Data Science}, vol.~1, no.~1, pp.
  46--77, 2019.

\bibitem{nam2013cosparse}
S.~Nam, M.~E. Davies, M.~Elad, and R.~Gribonval, ``The cosparse analysis model
  and algorithms,'' \emph{Applied and Computational Harmonic Analysis},
  vol.~34, no.~1, pp. 30--56, 2013.

\bibitem{le2015chasing}
L.~Le~Magoarou and R.~Gribonval, ``Chasing butterflies: In search of efficient
  dictionaries,'' in \emph{2015 IEEE International Conference on Acoustics,
  Speech and Signal Processing (ICASSP)}.\hskip 1em plus 0.5em minus
  0.4em\relax IEEE, 2015, pp. 3287--3291.

\bibitem{parker2014bilinear1}
J.~T. Parker, P.~Schniter, and V.~Cevher, ``Bilinear generalized approximate
  message passing—part i: Derivation,'' \emph{IEEE Transactions on Signal
  Processing}, vol.~62, no.~22, pp. 5839--5853, 2014.

\bibitem{parker2014bilinear2}
------, ``Bilinear generalized approximate message passing—part ii:
  Applications,'' \emph{IEEE Transactions on Signal Processing}, vol.~62,
  no.~22, pp. 5854--5867, 2014.

\bibitem{spielman2012exact}
D.~A. Spielman, H.~Wang, and J.~Wright, ``Exact recovery of sparsely-used
  dictionaries,'' in \emph{Conference on Learning Theory}, 2012, pp. 37--1.

\bibitem{jung2016minimax}
A.~Jung, Y.~C. Eldar, and N.~G{\"o}rtz, ``On the minimax risk of dictionary
  learning,'' \emph{IEEE Transactions on Information Theory}, vol.~62, no.~3,
  pp. 1501--1515, 2016.

\bibitem{cvxpy}
S.~Diamond and S.~Boyd, ``{CVXPY}: A {P}ython-embedded modeling language for
  convex optimization,'' \emph{Journal of Machine Learning Research}, vol.~17,
  no.~83, pp. 1--5, 2016.

\bibitem{mosek}
\BIBentryALTinterwordspacing
``{The MOSEK optimization software}.'' [Online]. Available:
  \url{http://www.mosek.com/}
\BIBentrySTDinterwordspacing

\bibitem{garg2009gradient}
R.~Garg and R.~Khandekar, ``Gradient descent with sparsification: an iterative
  algorithm for sparse recovery with restricted isometry property,'' in
  \emph{Proceedings of the 26th Annual International Conference on Machine
  Learning}.\hskip 1em plus 0.5em minus 0.4em\relax ACM, 2009, pp. 337--344.

\bibitem{dask}
\BIBentryALTinterwordspacing
{Dask Development Team}, \emph{Dask: Library for dynamic task scheduling},
  2016. [Online]. Available: \url{http://dask.pydata.org}
\BIBentrySTDinterwordspacing

\bibitem{sulam2019multi}
J.~Sulam, A.~Aberdam, A.~Beck, and M.~Elad, ``On multi-layer basis pursuit,
  efficient algorithms and convolutional neural networks,'' \emph{IEEE
  transactions on pattern analysis and machine intelligence}, 2019.

\bibitem{Kingma2014AdamAM}
D.~P. Kingma and J.~Ba, ``Adam: A method for stochastic optimization,'' in
  \emph{Proc. the 3rd International Conference on Learning Representations
  (ICLR)}, 2014, pp. 1--15.

\bibitem{chen2018theoretical}
X.~Chen, J.~Liu, Z.~Wang, and W.~Yin, ``Theoretical linear convergence of
  unfolded ista and its practical weights and thresholds,'' in \emph{Advances
  in Neural Information Processing Systems}, 2018, pp. 9061--9071.

\end{thebibliography}
\end{document}